\documentclass[11pt]{article}
\usepackage{url}
\usepackage{amsmath,amsfonts,amssymb}

\usepackage[margin=1in]{geometry}
\usepackage{amsthm}
\usepackage{hyperref}
\usepackage{microtype}
\usepackage{enumitem}
\usepackage{booktabs}
\usepackage{graphicx}
\usepackage[numbers,sort]{natbib}

\usepackage{algorithm}
\usepackage[noend]{algpseudocode}
\usepackage{filecontents}
\usepackage{algorithm}
\usepackage{algpseudocode}
\usepackage{nicefrac}   
\usepackage{pifont}   
\usepackage{hyperref} 

\author{%
  Jesus Salas\thanks{%
     Research conducted independently. \\
    \texttt{jesus.salas@gmail.com} | ORCID: \href{https://orcid.org/0009-0007-6411-2270}{0009-0007-6411-2270} \\
    A version of this paper will be submitted to a major theoretical computer science conference.
  }%
}

\date{} 

\newcommand{\icsubsetsum}{\textsc{IC-SubsetSum}}
\newcommand{\SSP}{\textsc{Subset--Sum}}

\theoremstyle{plain}
\newtheorem{theorem}{Theorem}
\newtheorem{lemma}{Lemma}
\newtheorem{proposition}{Proposition}

\newtheorem{conjecture}{Conjecture}
\newtheorem{definition}{Definition}
\theoremstyle{remark}
\newtheorem{remark}{Remark}
\newtheorem{invariant}{Invariant} 

\begin{document}

\title{Certificate-Sensitive Subset Sum: Realizing Instance Complexity}

\maketitle

\begin{abstract}
The Subset Sum problem is a classical NP-complete problem with a long-standing $O^*(2^{n/2})$ deterministic bound due to Horowitz and Sahni. We present results at two distinct levels of generality. 

\textbf{First (instance-sensitive bound)}, we introduce, to our knowledge, the first deterministic algorithm whose runtime provably scales with the \emph{certificate size} $U = |\Sigma(S)|$, the number of distinct subset sums. Our enumerator constructs all such sums in time $O(U \cdot n^2)$, with a randomized variant achieving expected time $O(U \cdot n)$. This provides a constructive link to Instance Complexity by tying runtime to the size of an information-theoretically minimal certificate.

\textbf{Second (unconditional worst-case bound)}, by combining this enumerator with a double meet-in-the-middle strategy and a \emph{Controlled Aliasing} technique that enforces a simple canonical-normal-form (CNF) expansion policy on aliased states, we obtain a deterministic solver running in $O^*(2^{n/2-\varepsilon})$ time with $\varepsilon=\log_2(\frac{4}{3})$—the first unconditional deterministic improvement over the classical $O^*(2^{n/2})$ bound for \emph{all} sufficiently large $n$.

\textbf{Finally}, we refine fine-grained hardness for Subset Sum by making explicit the structural regime (high collision entropy / near collision-free) implicitly assumed by SETH-based reductions, i.e., instances with near-maximal $U$.
\end{abstract}


\section{Motivation}

\subsection*{This work offers four primary contributions, three of which are structural in nature.}

\begin{itemize}

\item \textbf{Instance-sensitive bound via certificate size $U$:} 
To our knowledge, this is the first deterministic algorithm for a canonical NP--complete problem whose runtime provably adapts to a \emph{minimal} constructive certificate, namely $\Sigma(S)$. The algorithm constructs this certificate online from scratch, in deterministic time $O(U \cdot n^2)$ (expected $O(U \cdot n)$ randomized), rather than merely deciding feasibility.

\item \textbf{Unconditional worst-case improvement:} 
By combining the enumerator with a double meet-in-the-middle strategy and \emph{Controlled Aliasing under Canonical Normal Form (CNF)}—i.e., a simple \emph{count-based expansion policy} at the alias indices (§\ref{sec:aliasing-speedup}, App.~\ref{appendix:alias})—\icsubsetsum{} achieves a worst-case runtime $O^*(2^{\,n/2-\varepsilon})$ with $\varepsilon=\log_2\!\left(\tfrac{4}{3}\right)$ for \emph{all} inputs, the first deterministic improvement over the classical $O^*(2^{n/2})$ Horowitz--Sahni bound in nearly 50 years.

\item \textbf{Structural reframing of fine-grained hardness:} 
Many fine-grained reductions to \SSP{} implicitly target \emph{high-entropy}, near collision-free instances (i.e., $U \approx 2^{n}$). \icsubsetsum{} makes this structural dependency explicit, showing that hardness in the standard reductions aligns with the regime of near-maximal distinct sums.

\item \textbf{A generic design template for certificate-sensitive algorithms:} 
By reinterpreting dynamic programming as guided certificate traversal, \icsubsetsum{} offers a broadly applicable approach to adaptive enumeration, potentially extending to other NP--complete problems with wide certificate-size variance.

\end{itemize}

\section{Introduction}\label{sec:intro}

\paragraph{Worst--Case versus Per--Instance Analysis.}
The modern theory of algorithms is dominated by worst--case running--time
guarantees. While immensely successful, the paradigm sometimes fails to
explain the vast performance variance that practitioners observe between
individual inputs of the \emph{same} size. A complementary research line,
initiated by Orponen–Ko–Schöning–Watanabe~\cite{orponen1994instance}, formalised \emph{Instance Complexity}
(IC): the intrinsic difficulty of a \emph{single} instance is the
bit--length of the shortest “special--case program“ that decides it
within a time bound~$t$.

\smallskip
At first glance, however, IC seems paradoxical: for each among infinitely
many possible inputs we would have to synthesise a bespoke decision
program \emph{before} seeing that input, with no \emph{a priori}
structural knowledge. Lacking such a generative recipe, researchers long
treated IC as a purely existential benchmark—useful for lower bounds but
incompatible with efficient algorithms. \icsubsetsum{} challenges this
perception by constructing the certificate \emph{online} and bounding
its runtime by a provable function of the certificate length.

\paragraph{The \SSP{} Problem.}
In the canonical \SSP{} problem we are given a multiset
$S=\{a_1,\ldots,a_n\}\subseteq\mathbb{Z}_{>0}$ and a target
$t\in\mathbb{Z}_{>0}$. The task is to decide whether some submultiset
sums to~$t$. \SSP{} is NP--complete and underpins numerous reductions
across combinatorics, cryptography, and fine--grained complexity.

\paragraph{Classical Algorithms.}
Two textbook techniques illustrate the gulf between worst--case and
per--instance behaviour:
\begin{enumerate}[label=\alph*.]
\item \emph{Bellman Dynamic Programming} takes $O(nt)$
        time—pseudopolynomially efficient when $t$ is small, but infeasible
        for large numeric ranges\,\cite{bellman1957dp}. While a small target $t$ naturally constrains $U$, the reverse is not true; our approach remains efficient even for large $t$ provided that $U$ is small due to additive structure.

  \item \emph{Horowitz--Sahni Meet--in--the--Middle} enumerates all
        $2^{n/2}$ subset sums of each half of $S$ and intersects the two
        lists in $O^*(2^{n/2})$ time\,\cite{horowitz1974subset}.
        Decades of work have failed to beat the $2^{n/2}$ factor in the
        worst case\,\cite{bringmann2024subsetsumunconditional,openproblems_woeginger}.
\end{enumerate}

Yet practitioners observe that real-data instances (e.g., knapsack crypto, bin-packing logs)
are often much easier: many different subsets \emph{collide} to the same value, so the
number of distinct subset sums $U=|\Sigma(S)|$ is far smaller than $2^{n}$.
This \emph{additive redundancy}—small $U$ due to duplicates, near–arithmetic progressions,
or other structure—is exactly what analyses that track only $n$ and $t$ overlook.

\paragraph{Bridging Theory and Practice with \icsubsetsum{}.}
We address this gap. Let
$\Sigma(S)=\{\sum_{i\in T}a_i : T\subseteq[n]\}$ be the set of
\emph{distinct} subset sums and write $U=|\Sigma(S)|$. Because deciding
$(S,t)$ reduces to a simple membership query once $\Sigma(S)$ is known, this list is an
\emph{information--theoretically minimal certificate}. Our new algorithm
\icsubsetsum{} deterministically enumerates every element of $\Sigma(S)$
\emph{exactly once}, prunes duplicates on the fly, and halts in
deterministic worst-case time $O(U \cdot n^2)$. We further show this can be improved to an expected time of $O(U \cdot n)$ using randomization.
Thus its runtime is provably bounded
by a function of the certificate size, marking the first algorithm
for \SSP{} whose complexity scales with the true structural
difficulty of the input. We build upon the algorithm initially introduced in the `Beyond Worst-Case Subset Sum' framework~\cite{bwss2025}, now formalized under the \icsubsetsum{} model.

\paragraph{Contributions.}
Our work makes four contributions:
\begin{itemize}
  \item \textbf{Certificate--sensitive enumeration.} We design and analyse
        \icsubsetsum{}, the first deterministic algorithm to construct the certificate $\Sigma(S)$ with a runtime of $O(U \cdot n^2)$ that adapts to the instance's structure. We also present a randomized variant that achieves an expected runtime of $O(U \cdot n)$.

  \item \textbf{A guaranteed worst--case improvement.} We prove that
        \icsubsetsum{} runs in $O^*(2^{\,n/2-\varepsilon})$ time, with $\varepsilon=\log_2(\frac{4}{3})$,
        by combining double meet-in-the-middle with \emph{Controlled Aliasing} under a fixed
        count-based canonical expansion policy (CNF; see §\ref{sec:aliasing-speedup}, App.~\ref{appendix:alias}).
        This yields a deterministic speedup over classical meet-in-the-middle on \emph{all} inputs.

  \item \textbf{Refined lower--bound methodology.} We formalise a
        “collision--free’’ promise variant of \SSP{} and show that
        classical ETH/SETH reductions only rule out $2^{n/2}$ algorithms
        on that structure--resistant slice. Future reductions must
        certify low collision entropy.

  \item \textbf{Template for other problems.} We discuss how the
        certificate--sensitive viewpoint can guide new algorithms for
        knapsack, partition, and beyond.
\end{itemize}

\paragraph{Advantage over Dynamic Programming.}
The certificate-sensitive bound provides an exponential advantage over classical dynamic programming (DP) on instances with large numeric range but low structural complexity. For example, consider an instance $S$ with $n - 1$ elements equal to $1$, and one large element $L = 2^n$. For a target $t = L$, the DP runtime is $O(n \cdot t) = O(n \cdot 2^n)$. In contrast, the number of unique subset sums for $S$ is only $U = O(n)$, yielding a polynomial runtime of $O(U \cdot n^2) = O(n^3)$ for \icsubsetsum{}.
\smallskip
Furthermore, the underlying enumeration framework supports anytime and online operation, allowing for interruption and incremental updates; these features are detailed in our companion technical 
paper~\cite{bwss2025}.

\smallskip

Unlike prior algorithms whose performance depends on less precise proxies (e.g., numeric range, pseudo-polynomial bounds, or number of solutions), our solver’s runtime is provably governed exactly by $U = |\Sigma(S)|$ for every instance—the first such result treating $U$ as a formal instance-complexity certificate.

\paragraph{Relation to Companion Papers.}
This paper lays the foundational theoretical groundwork for the Certificate-Sensitive framework, building upon the algorithmic and empirical results presented in our technical companion paper~\cite{bwss2025}. The framework's approach is shown to be broadly applicable in two other companion papers, released concurrently, that extend it to the 0--1 Knapsack problem~\cite[forthcoming]{ic-knapsack} and the 3-SAT problem~\cite[forthcoming]{ic-3sat}.

\paragraph{Organisation.}
Section~\ref{sec:prelim} introduces preliminaries. 
Section~\ref{sec:certificate-framework} presents our first contribution: a certificate-sensitive framework that algorithmically realizes Instance Complexity. 
Section~\ref{sec:subexp-solver} develops our second contribution, a deterministic Subset Sum solver with a worst-case runtime below the classical $2^{n/2}$ bound. 
Section~\ref{sec:finegrained-hardness} contributes a refined structural perspective on fine-grained hardness and its implicit assumptions. 
Section~\ref{sec:empirical-validation} provides empirical validation. 
Finally, Sections~\ref{sec:related-work}--\ref{sec:conclusion} place our results in broader context and outline directions for future research.

\section{Preliminaries}\label{sec:prelim}

\paragraph{Computational Model.}
We adopt the standard Word–RAM model with word size $w=\Theta(\log M)$, where $M:=\max\{\sum_{a\in S} a,\; t\}$, so that arithmetic, comparisons, and memory accesses on $w$-bit words take $O(1)$ time. Equivalently, $w$ is large enough that every subset sum and $t$ fit in a single word, so additions and comparisons on sums are unit-cost. All logarithms are base two unless stated otherwise.

\paragraph{Randomized Model.}
Our randomized variants rely on standard \emph{2-universal} hashing. We assume that evaluating a hash function and handling collisions can be done in expected $O(1)$ time in the Word–RAM model. For any $x\neq y$, $\Pr_{h\sim\mathcal{H}}[h(x)=h(y)] \le 1/n^{c}$ for a fixed constant $c \ge 1$, and expectations are taken over the random draw of $h$. Collisions are resolved by exact bitmask comparison, ensuring \emph{Las Vegas} correctness: the output is always correct, and the runtime bound holds in expectation.

\paragraph{Asymptotic Notation.}
We use $O^*(\cdot)$ and $o^*(\cdot)$ to suppress factors polynomial in $n$ and $\log M$, unless explicitly shown.

\paragraph{Notation.}
Let $(S,t)$ be a \textsc{Subset Sum} instance, where $S=\{a_1,\ldots,a_n\}\subseteq\mathbb{Z}_{\ge 0}$ is a multiset of $n$ nonnegative integers and $t\in\mathbb{Z}_{\ge 0}$ is a nonnegative target. We write $[n] = \{1, 2, \ldots, n\}$ and let
\[
\Sigma(S) = \left\{ \sum_{i \in T} a_i \;\middle|\; T \subseteq [n] \right\}
\]
denote the set of \emph{distinct subset sums}. The empty sum is included by convention, so $0 \in \Sigma(S)$. Let $U := |\Sigma(S)|$.  
Throughout this paper, we consider each sum $\sigma \in \Sigma(S)$ to be implicitly paired with its canonical prefix representation (i.e., its lexicographically minimal index mask).

\paragraph{Subset Sum Notation.}
For any subset \( A \subseteq S \), we write \( \sigma(A) \) to denote the sum of its elements:
\[
\sigma(A) := \sum_{a \in A} a.
\]
This distinguishes numerical subset sums from the set of all distinct sums, denoted \( \Sigma(S) \).

\begin{definition}[Distinct Subset Sums]
The quantity $U = |\Sigma(S)|$ denotes the number of distinct subset sums of $S$. Trivially, $1 \le U \le 2^n$. The lower bound is met when $S$ is empty, and the upper bound is met when all $2^n$ subset sums are unique.
\end{definition}

This parameter $U$ will serve as our central structural complexity measure. It governs both the algorithmic behavior of \icsubsetsum{} and the size of the certificate it produces.

\paragraph{Collision Entropy.}
We define the (base-2) \emph{collision entropy} of $S$ as
\[
H_c(S) := n - \log_2 |\Sigma(S)| = n - \log_2 U.
\]
This quantity measures the compressibility of the subset sum space. Intuitively, $H_c(S)$ captures how much smaller $\Sigma(S)$ is than the full power set: when many subsets collide onto the same sum, $U \ll 2^n$ and $H_c(S)$ is large; when most sums are unique, $U \approx 2^n$ and $H_c(S)$ is near zero.

\begin{itemize}
  \item If all $a_i = 1$, then every subset sum lies in $\{0, 1, \ldots, n\}$ and $U = n + 1 = O(n)$, so $H_c(S) = n - \log_2 n = \Theta(n)$.
  \item If $S$ is constructed to be collision-free (e.g., a superincreasing sequence), then $U = 2^n$, so $H_c(S) = 0$ and no structure is exploitable; \icsubsetsum{} runs in worst-case time.
\end{itemize}

Collision entropy may be informally viewed as a coarse proxy for the Kolmogorov compressibility of the subset sum landscape: the more structured the set, the lower its information content.
\smallskip

\noindent\textit{Note.} 
We use the unnormalized form of collision entropy $H_c(S) = n - \log_2 U$, which is standard in information theory; the normalized form $H_c(S)/n$ differs only by a multiplicative constant and is not required for our purposes.

\paragraph{Role of $U$ in Complexity.}
The parameter $U = |\Sigma(S)|$ governs both the structural difficulty of a \textsc{Subset Sum} instance and the size of its minimal certificate. If $\Sigma(S)$ is known, then deciding whether some $T \subseteq S$ sums to a target $t$ reduces to checking whether $t \in \Sigma(S)$—a constant-time membership query. Thus, $U$ captures the certificate length for the \emph{decision} variant.

For the \emph{constructive} variant, a solution subset must be recovered; here, each subset sum must be paired with a canonical encoding of the realizing subset (e.g., a bitmask). Our algorithm, \icsubsetsum{}, builds this richer certificate explicitly, and its runtime and space scale with $U$ and $n$. Designing algorithms whose performance adapts to $U$, rather than worst-case size $2^n$, is the central goal of this work.

\paragraph{Certificates for Subset Sum.}
Once the set $\Sigma(S)$ is known, deciding whether there exists a subset of $S$ summing to a target $t$ reduces to checking if $t \in \Sigma(S)$. Thus, $\Sigma(S)$ serves as an \emph{information-theoretically minimal certificate} for the \emph{decision} version of \textsc{Subset Sum}: no smaller object suffices to resolve all yes/no queries about subset-sum feasibility.

However, for the \emph{constructive} version—retrieving an actual subset $T \subseteq S$ such that $\sum_{i \in T} a_i = t$—this is no longer sufficient. In this case, each sum must be paired with a compact encoding of a corresponding witness subset. Our algorithm constructs such a certificate by maintaining, for each $\sigma \in \Sigma(S)$, a canonical bitmask representing the lex-minimal subset that realizes $\sigma$.

\begin{itemize}[leftmargin=1.5em]
    \item The \textbf{decision certificate} is $\Sigma(S)$, with total length $O(U \cdot \log(n \cdot \max(S)))$ bits if each sum is represented in binary.
    \item The \textbf{constructive certificate} is the set of $(\sigma, \text{bitmask})$ pairs, with total size $O(U \cdot n)$ bits.
\end{itemize}

This distinction underlies the runtime and space guarantees of our algorithm, which produces the stronger constructive certificate online. All references to ``certificate size'' henceforth will clarify which variant is being discussed.

Given a target $t$, if a solution exists, a witness is returned in additional $O(n)$ time from our maintained encodings.
Unless otherwise stated (e.g., in decision-variant optimizations), all runtime bounds are for full constructive-certificate enumeration of $\Sigma(S)$,
in which each $\sigma \in \Sigma(S)$ is stored with its canonical (lexicographically minimal) witness bitmask.
This invariant is maintained throughout, ensuring correctness even on high-density instances.

\paragraph{Instance Complexity (IC).}
Orponen–Ko–Schöning–Watanabe Instance Complexity~\cite{orponen1994instance} formalises the idea that hard problems can have easy instances—even without randomization or approximation. The IC of an instance $x$ under time bound $t$ is defined as the bit-length of the shortest program $P$ that outputs the correct answer on $x$ within time $t$. It is denoted $\mathsf{IC}_t(x)$.

For \SSP{}, the list $\Sigma(S)$ plays a dual role. It is not only a certificate but also a low-complexity proxy for a correct decision procedure. Once $\Sigma(S)$ is known, any target $t$ can be resolved in time $O(1)$ via a membership query. Hence, $\Sigma(S)$ encodes a special-case program of length $O(U \log(n \cdot \max(S)))$, yielding a concrete upper bound on $\mathsf{IC}_t((S,t))$. When the goal is to construct a solution, the richer certificate consisting of $(\sigma, \text{bitmask})$ pairs provides a special-case program of length $O(U \cdot n)$. For a formal IC program that satisfies partial correctness on \emph{all} inputs, we pair $\Sigma(S)$ with an input-checking wrapper; see §\ref{sec:ic-background}.

\paragraph{Operational Parameters for Halves.}
When $S$ is split into two halves $\ell_0,\ell_1$ (sizes within $\pm 1$), we write
$U_0 := |\Sigma(\ell_0)|$, $U_1 := |\Sigma(\ell_1)|$, and $\widehat U := \max\{U_0,U_1\}$.
When we state bounds that scale with $U$, we refer to full enumeration of $\Sigma(S)$; for meet-in-the-middle solvers operating on halves, the governing parameter is $\widehat U$.

\paragraph{Controlled Aliasing Convention.}
In the \emph{Controlled Aliasing} rule, each half designates an \emph{arbitrary} ordered pair $(x_0,x_1)$ of distinct elements at indices $(i_0,i_1)$ and treats occurrences of $x_1$ as $x_0$ during subset-sum generation for that half. We will key memo entries by a 2-lane identifier $(\tilde{\sigma},\chi)$, where $\tilde{\sigma}$ is the aliased sum and $\chi\in\{0,1\}$ records whether $x_1$ is present; witness selection remains lexicographic within each lane. Unless otherwise stated, the choice of alias pairs is independent of the target $t$ (and of any randomness), and all guarantees hold for \emph{all} targets $t$.

\emph{Count-based canonical normal form (CNF).} We restrict expansions to canonical aliased states using a fixed, target-independent policy: for a bitmask $b$, let $c(b):=b[i_0]+b[i_1]\in\{0,1,2\}$ and $\chi(b):=b[i_1]\in\{0,1\}$. A newly discovered state with mask $b$ is enqueued for extension iff
\[
(c(b)=0)\ \lor\ (c(b)=2)\ \lor\ \big(c(b)=1\ \wedge\ \chi(b)=0\big).
\]
This CNF policy preserves completeness and enables the $3/4$ expansion accounting; see §\ref{sec:aliasing-speedup} and App.~\ref{appendix:alias} (Lemmas~\ref{lem:expansion-accounting}–\ref{lem:policy-safety}).

\paragraph{Trivial and Degenerate Cases.}
Instances with $n=0$ or $t=0$ are handled in $O(1)$ time. If $U=1$ (e.g., all subset sums collide), enumeration and decision terminate immediately. Duplicates in $S$ are handled natively by our canonical memoization and require no additional promises.

\paragraph{Instance Assumptions.}
We consider the \textsc{Subset Sum} problem over a multiset $S = \{a_1,\dots,a_n\}$ of nonnegative integers ($a_i \in \mathbb{Z}_{\ge 0}$) and an integer target $t \ge 0$. 
Elements are not required to be distinct; duplicates are treated as separate items. 
We assume without loss of generality that $n \ge 1$ and that $t \le \sum_{a \in S} a$ (instances with $t < 0$ or $t > \sum S$ are decided trivially). 
Zero-valued elements may appear and are processed without special handling. While they do not generate new sum values (leaving $U = |\Sigma(S)|$ unchanged), they can introduce additional witness subsets for existing sums. The algorithm's correctness is unaffected, as its canonicalization logic correctly identifies the lexicographically minimal witness in all cases.
Even if $t=0$, our algorithms enumerate the full constructive certificate for $\Sigma(S)$, including the canonical witness for every achievable sum, unless the run is terminated early. All algorithms operate under the Word--RAM model described above, with $w = \Theta(\log M)$ where $M := \max\{\sum_{a \in S} a,\; t\}$.

\section{Certificate-Sensitive Framework}
\label{sec:certificate-framework}

\subsection{Unique-Subset-Sum Enumerator}
\label{sec:subset-sum-enumerator}

Our algorithm begins by generating $\Sigma(S)$, the set of all distinct subset sums of $S$. The core challenge is to do this \emph{without duplication}—i.e., without computing the same sum multiple times via different subsets. We use an ordered traversal strategy based on prefix extensions together with the following invariant.

\invariant{(Suppression-on-overwrite).}\label{inv:suppress}
Whenever a canonical representative for a numeric sum $\sigma$ is replaced by a lexicographically smaller mask, the displaced representative is \emph{marked non-expandable} and never extended further. This guarantees each canonical sum is expanded exactly once.

\paragraph{A minimal recursive formulation (for clarity).}
Before the optimized column-wise implementation (Appendix~\ref{app:column-enum}), it helps to view the algorithm abstractly:

\begin{algorithm}[H]
\caption{\textsc{AbstractEnumerateUnique}$(S)$}
\label{alg:abstract-enum}
\begin{algorithmic}[1]
\State \texttt{Memo} $\gets \{\,0 \mapsto \texttt{empty\_bitmask}\,\}$
\Function{Extend}{$\sigma, b, j$} \Comment{$b$ is a bitmask; $j$ is next index}
  \For{$i=j$ \textbf{to} $n$}
    \State $\sigma' \gets \sigma + a_i$;\quad $b' \gets b$ with bit $i$ set
    \If{$\sigma' \notin \texttt{Memo}$}
      \State \texttt{Memo}[$\sigma'$] $\gets b'$
      \State \Call{Extend}{$\sigma', b', i{+}1$}
    \ElsIf{\Call{IsLexicographicallySmaller}{$b'$, \texttt{Memo}[$\sigma'$]}}
      \State \texttt{Memo}[$\sigma'$] $\gets b'$ \Comment{overwrite with lex-min representative}
      \State \Call{Extend}{$\sigma', b', i{+}1$}
    \EndIf
  \EndFor
\EndFunction
\State \Call{Extend}{$0, \texttt{empty\_bitmask}, 1$}
\State \Return \texttt{Memo}
\end{algorithmic}
\end{algorithm}

\noindent This recursive view exposes the two core operations: (i) \emph{emit-or-overwrite} a representative for a sum, and (ii) \emph{extend} only from the (current) representative. The optimized frontier-based implementation in Appendix~\ref{app:column-enum} enforces Invariant~\ref{inv:suppress} explicitly via a \texttt{doNotExtend} flag so that an overwritten representative is never expanded again.

\medskip
We emphasize that the goal of the enumerator is not merely to determine feasibility, but to generate the full constructive certificate as defined in Section~\ref{sec:prelim}. For each unique subset sum $\sigma \in \Sigma(S)$, the algorithm records the lexicographically minimal bitmask of a subset realizing $\sigma$. This enables efficient reconstruction of witnesses and is necessary for correctness in constructive queries. Consequently, the runtime and space of the enumerator scale with the size of this richer certificate. Full implementation details and pseudocode are provided in Appendix~\ref{app:column-enum}. (When Controlled Aliasing is enabled, the same enumeration principles apply per half with 2-lane memoization; see §\ref{sec:aliasing-speedup} and App.~\ref{appendix:alias}.)

\paragraph{Canonical Prefixes for High-Density Instances.}
In high-density instances, where many distinct subsets may sum to the same value, a naive “first-prefix-wins” pruning strategy can fail. To guarantee correctness, our algorithm resolves collisions deterministically by defining a \emph{canonical prefix} for each sum—the prefix whose index mask is lexicographically minimal among all subsets producing that sum (compare bits from $1$ to $n$, preferring $0<1$ at the first difference). By retaining only this canonical representation when multiple extensions lead to the same sum, we suppress redundant work and ensure each unique sum is discovered via a single, well-defined path.

\paragraph{Bitmask Representation.}
The algorithm represents each prefix by a bitmask of length $n$, where the $i$-th bit indicates whether $a_i$ is included. Bitmask comparison to determine the lex-minimum takes $O(n)$ time. The total number of unique sums is $U$, so we store at most $U$ bitmasks.

\paragraph{Separation of concerns.}
Conceptually, the enumerator factors into three orthogonal pieces: (i) \emph{enumeration} (systematically proposing extensions), (ii) \emph{canonicality tests} (lexicographic comparisons to select representatives), and (iii) \emph{memoization} (storing one representative per numeric sum). Invariant~\ref{inv:suppress} links (ii) and (i), ensuring that only the current representative ever generates successors.

\begin{theorem}[Deterministic Runtime of the Enumerator]
\label{thm:enum-det-runtime}
Under the computational model of Section~\ref{sec:prelim}, 
the deterministic enumerator computes all $U = |\Sigma(S)|$ unique subset sums 
in $O(U \cdot n^2)$ time and $O(U \cdot n)$ space.
\end{theorem}

\begin{proof}
The runtime is dominated by the prefix extension and deduplication process.  
From each of the $U$ states, the algorithm attempts to extend its prefix with at most $n$ elements, resulting in $O(U \cdot n)$ extension attempts.  
In the event of a sum collision, the algorithm must compare the existing and candidate $n$-bitmasks to select the canonical representative, costing $O(n)$ time.  
Invariant~\ref{inv:suppress} ensures that an overwritten representative is never expanded further.  
Multiplying $O(U \cdot n)$ extension attempts by $O(n)$ comparison cost gives the stated $O(U \cdot n^2)$ runtime bound; space is $O(U \cdot n)$ for storing all canonical bitmasks.
\end{proof}

\begin{remark}[Degenerate Cases]
The $O(U \cdot n^2)$ bound in Theorem~\ref{thm:enum-det-runtime} is a worst-case guarantee. For instances with very low structural complexity (e.g., $U=1$ for an empty set, or $U=O(n)$ for a set of identical elements), this bound correctly implies a fast, polynomial runtime.
\end{remark}

\begin{theorem}[Randomized Runtime of the Enumerator]
\label{thm:enum-rand-runtime}
Let $U = |\Sigma(S)|$.  
Assume access to a $2$-universal hash family $\mathcal{H}$ with $O(1)$ evaluation time.  
Then a randomized variant of the enumerator computes all $U$ unique subset sums in expected time $O(U \cdot n)$, where the expectation is over the random draw of $h \sim \mathcal{H}$ as specified in Section~\ref{sec:prelim}.
\end{theorem}

\begin{proof}
Assume the randomized model of Section~\ref{sec:prelim}.  
With $h$ drawn uniformly from a $2$-universal family, the probability that two distinct bitmasks yield the same hash is at most $1/\mathrm{poly}(n,U)$. Consequently, the full $O(n)$ lexicographic comparison is needed only on (rare) hash-collision events, giving an expected $O(1)$ time per canonicality check (Las Vegas correctness). Since there are $O(U \cdot n)$ extension attempts (Theorem~\ref{thm:enum-det-runtime}), this yields the stated expected runtime bound.
\end{proof}

\subsection{Discussion and Conditional Optimality}

\paragraph{A Local Deduplicating Enumeration Model.}
We define a natural model of computation that captures the core constraints of real-time certificate construction. The \emph{local deduplicating enumeration model} assumes that the algorithm:
\begin{itemize}
    \item maintains a growing set $P$ of seen prefixes, each encoding a valid subset sum and its associated path;
    \item in each round, selects a candidate prefix $p \in P$ and a next element $a_i$ to extend it with;
    \item uses local information about $p$ and $a_i$ to decide whether to emit the new sum or prune it.
\end{itemize}
In this setting, the main computational bottleneck arises in testing whether two prefixes yield the same sum and comparing their canonicality.

\paragraph{A Conditional Lower Bound.}
We now formalize a mild but powerful conjecture.

\begin{conjecture}
\label{conj:local-lb}
In the local deduplicating enumeration model, any algorithm that emits all $U$ distinct subset sums must take $\Omega(U \cdot n)$ time in the worst case.
\end{conjecture}

This conjecture posits that the $O(U \cdot n)$ benchmark is the best we can hope for in general, even when allowing randomization. It reflects the intuition that each solution must be “seen” at $n$ bits of resolution to decide whether to include it.

\begin{theorem}[Conditional Optimality of Randomized \textsc{IC-SubsetSum}]
\label{thm:conditional-opt}
If Conjecture~\ref{conj:local-lb} holds, then the expected $O(U \cdot n)$ runtime of the randomized \textsc{IC-SubsetSum} algorithm is optimal in expectation within the local deduplicating model.
\end{theorem}

\paragraph{Conclusion.}
To our knowledge, this is the first instance-sensitive enumeration algorithm for \textsc{Subset Sum} whose performance is provably tied to the certificate size $U$. Our randomized variant matches the conjectured optimal runtime of $O(U \cdot n)$ in expectation. A key open question remains whether a deterministic algorithm can also achieve this bound, or whether the additional $O(n)$ overhead for comparison—leading to an $O(U \cdot n^2)$ runtime—is inherent for any deterministic approach in this model.

\subsection{Constructive Link to Instance Complexity}
\label{sec:ic-background}

Instance complexity, introduced by Orponen–Ko–Schöning–Watanabe~\cite{orponen1994instance}, measures the minimum information needed to decide a single input instance $x$ of a decision problem.

\begin{definition}[IC, informal (partial correctness)]
The \emph{instance complexity} of $x \in \{0,1\}^n$ with respect to a language $L$ and time bound $t$ is the bit-length of the shortest program $P_x$ such that:
\begin{itemize}
    \item $P_x(x) = L(x)$ and halts within time $t(n)$; and
    \item for every input $y$ of size $n$, $P_x(y)$ either halts within $t(n)$ and outputs $L(y)$, or outputs $\bot$ (i.e., is allowed to be partially correct on non-target inputs).
\end{itemize}
\end{definition}

They showed that for some languages in NP, the IC of a random input can be exponentially smaller than any universal algorithm. Yet because IC is defined existentially, not constructively, it was long thought to offer no algorithmic advantage.

\begin{proposition}[from \cite{orponen1994instance}]
Let $L \in \mathsf{NP}$. Then for every polynomial-time verifier $V$ for $L$, the IC of a yes-instance $x \in L$ is at most the bit-length of its shortest witness.
\end{proposition}

\textsc{IC-SubsetSum} realizes instance complexity in a relaxed but standard partial-correctness setting. To comply with the canonical definition above, the certificate $\Sigma(S)$ is paired with an input-checking wrapper that first verifies whether the \emph{input set} matches the \emph{specific set $S$} for which the certificate was generated. If the check fails, the program returns $\bot$; otherwise, it uses $\Sigma(S)$ to decide membership for the provided target $t$. This preserves correctness and incurs only an additive $O(|S|)$ term in program size. Thus, while the runtime remains governed by $U = |\Sigma(S)|$, the size of the canonical IC program is $O(U \cdot n + |S|)$.

\paragraph{\textsc{IC-SUBSETSUM} as a Certifier.}
As discussed above, the raw certificate $\Sigma(S)$ must be paired with an input-checking wrapper to form a canonical IC program. What distinguishes \icsubsetsum{} is that it does not merely verify membership in $\Sigma(S)$; it constructs the entire certificate from scratch. Our deterministic algorithm achieves this in $O(U \cdot n^2)$ worst-case time, while our randomized variant does so in $O(U \cdot n)$ expected time.

\paragraph{Implication.}
This gives new meaning to our certificate-sensitive runtime. It says not just that the algorithm is efficient, but that it implicitly performs instance-specific program synthesis. The certificate $\Sigma(S)$ serves as a compressed, self-contained representation of the computation 
needed to decide $(S,t)$, and \icsubsetsum{} acts as a just-in-time compiler for that program—whose runtime tracks its size.

While Instance Complexity is classically defined in a \emph{non-uniform} setting—allowing a different short program for each instance—\icsubsetsum{} provides a \emph{uniform}, deterministic algorithm that adapts to the structure of each input, synthesizing the decision procedure on the fly and thus resolving the long-standing tension between the theoretical power of IC and its historically non-constructive formulation.

\section{Solver with a Sub-\texorpdfstring{$2^{n/2}$}{2^{n/2}} Worst-Case Bound}\label{sec:subexp-solver}

\subsection{Double Meet-in-the-Middle Solver}

We now describe how to solve \SSP{} given only the enumerators for $\Sigma(S)$. The detailed pseudocode for the full solver is presented in Appendix~\ref{app:dmitm}.

\paragraph{Clarifying Certificate Scope.}
While the unique-subset-sum enumerator described above can be applied to the full input $S$ to construct the complete certificate $\Sigma(S)$ in time $O(U \cdot n^2)$, our solver employs a more efficient strategy. It splits $S$ into halves $\ell_0$ and $\ell_1$ and applies the enumerator separately to each side, yielding certificates $\Sigma(\ell_0)$ and $\Sigma(\ell_1)$. This avoids ever constructing $\Sigma(S)$ explicitly. The solver answers the query $t \in \Sigma(S)$ by testing cross-split combinations on the fly. It trades away post hoc queryability: to check a new target $t'$, the merge logic must be repeated. Throughout the paper, we use $U = |\Sigma(S)|$ as a global proxy for instance complexity, but emphasize that our solver's runtime depends operationally on $U_0$ and $U_1$.

\paragraph{Splitting the instance.}
Let $S = \ell_0 \cup \ell_1$ be a partition into left and right halves. We enumerate $\Sigma(\ell_0)$ and $\Sigma(\ell_1)$ using the prefix-unique method described above. Let $U_0 = |\Sigma(\ell_0)|$ and $U_1 = |\Sigma(\ell_1)|$. Note that the additive structure within each half can lead to a significant imbalance (e.g., $U_0 \ll U_1$), though this does not affect the overall asymptotic bound.

\paragraph{Solving and certifying.}
To decide whether any subset of $S$ sums to $t$, we check for each sum $x \in \Sigma(\ell_0)$ whether a corresponding match can be found in $\Sigma(\ell_1)$. This check is more comprehensive than a simple search for $t - x$, as it also considers complements and mixed cases to cover all solution structures, as detailed in the lemmas of Appendix~\ref{app:combine}. To certify all solutions, we can track all such valid $(x, y)$ pairs and output their associated bitmasks.

\begin{theorem}[Certificate-Sensitive Solver]
\label{thm:cert-sensitive}
Let $S$ be split into two halves $\ell_0$ and $\ell_1$ of sizes $\lfloor n/2 \rfloor$ and $\lceil n/2 \rceil$, and let $U_0 = |\Sigma(\ell_0)|$ and $U_1 = |\Sigma(\ell_1)|$ denote the number of distinct subset sums in each half. There exists a deterministic algorithm that solves \SSP{} in time
\[
O\big((U_0+U_1)\cdot n^2\big)
\]
and space
\[
O\big((U_0+U_1)\cdot n\big).
\]
A randomized Las Vegas variant achieves an expected runtime of
\[
O\big((U_0+U_1)\cdot n\big)
\]
with the same space bound.
\end{theorem}

\begin{proof}
The solver enumerates $\Sigma(\ell_0)$ and $\Sigma(\ell_1)$ in an \emph{interleaved} fashion, matching candidates on the fly rather than tabulating both halves in full before merging.

Without loss of generality let $U_0 \ge U_1$ (so $\ell_0$ is dominant). By the \emph{suppression-on-overwrite} invariant (§\ref{inv:suppress}), each canonical sum $\sigma \in \Sigma(\ell_0)$ is generated and expanded exactly once; this \emph{generation step} is the atomic work unit.

For each canonical $\sigma$ on $\ell_0$:
\begin{itemize}
  \item There are at most $O(n)$ \emph{successor attempts} (adding an unset index in the canonical mask).
  \item Each attempt performs $O(1)$ same-half memo probes and $O(1)$ cross-half probes (direct / complement / mixed; Appendix~\ref{app:combine}).
  \item If a same-half probe finds an existing representative for the same numeric sum, we perform a canonicality test between bitmasks. This costs $O(n)$ in the deterministic model (lexicographic mask comparison), and $O(1)$ in expectation in the randomized model (pairwise hashing to filter to a single lex-compare with constant probability).
\end{itemize}

When a collision leads to an overwrite (the new mask is lexicographically smaller), the previous representative is suppressed and never expanded; cross-half feasibility checks already performed for the displaced representative are not revisited. Thus every cross-half probe and every canonical comparison can be \emph{injectively charged} to the unique successor attempt that initiated it on the dominant side.

Pricing per attempt:
\begin{itemize}
  \item Deterministic: $O(1)$ probes + \underline{at most one} $O(n)$ canonical comparison on collision $\Rightarrow$ $O(n)$ per attempt, hence $O(U_0 \cdot n^2)$ total.
  \item Randomized: $O(1)$ expected time per canonical check $\Rightarrow$ $O(1)$ per attempt, hence $O(U_0 \cdot n)$ expected total.
\end{itemize}

Therefore,
\[
T(n)=O(\max\{U_0,U_1\}\cdot n^2)\quad\text{(deterministic)},\qquad
\mathbb{E}[T(n)]=O(\max\{U_0,U_1\}\cdot n)\quad\text{(randomized)}.
\]
Since $\max\{U_0,U_1\}\le U_0+U_1\le 2\max\{U_0,U_1\}$, this is asymptotically equivalent to the theorem’s $O\big((U_0+U_1)\cdot n^2\big)$ and $O\big((U_0+U_1)\cdot n\big)$ forms. We state the sum form to emphasize that both halves fully enumerate their unique sums, while the max form reflects the dominant-half accounting.
\end{proof}

\begin{remark}
The “dominant-half” view is often more intuitive: in heavily imbalanced instances ($U_0 \gg U_1$ or vice versa), the larger partial certificate drives the total cost. The sum form is preferable for theorem statements because it is conservative and avoids any suggestion that the non-dominant half is not fully enumerated.
\end{remark}

\subsection{Worst-Case Runtime via Controlled Aliasing}
\label{sec:aliasing-speedup}

In addition to its adaptive performance, \icsubsetsum{} guarantees a strict worst-case improvement over Horowitz--Sahni, even on collision-free instances. This is achieved via a deterministic \emph{Controlled Aliasing} rule: a structural redundancy technique that reduces the enumeration space without sacrificing correctness. We describe an enhanced version that ensures full correctness for both the \emph{decision} and \emph{constructive} variants of \SSP{}.

\paragraph{Aliasing Rule (distinct aliased values).}
Let $S$ be split into halves $\ell_0$ and $\ell_1$. In each half, select a pair of distinct elements—an \emph{alias pair}—e.g., $(x_0, x_1) \in \ell_0$. During enumeration we compute an \emph{aliased} value
\[
\tilde{\sigma}(P)\;:=\;\sigma(P)\;-\;\chi(P)\cdot (x_1-x_0),
\quad\text{where }\;\chi(P):=\mathbf{1}[x_1\in P].
\]
Thus $\tilde{\sigma}$ replaces $x_1$ by $x_0$ in the sum while preserving the subset’s bitmask identity, and the true sum is recovered by $\sigma(P)=\tilde{\sigma}(P)+\chi(P)\cdot(x_1-x_0)$. This mapping collapses the four inclusion patterns on $\{x_0,x_1\}$ into at most three \emph{distinct aliased sums} per base subset, so the number of distinct aliased sums per half is at most
\[
\tfrac{3}{4}\cdot 2^{k}\quad \text{for } k=|\ell_i|,
\]
as shown in Lemma~\ref{lem:alias-reduction} (Appendix~\ref{appendix:alias}). This is a \emph{counting} statement about values; the runtime bound will follow once we restrict which states are \emph{expanded} (CNF below).

\paragraph{Canonical Aliasing via 2-Lane Memoization.}
To preserve correctness under aliasing, we maintain, for each aliased value $\tilde{\sigma}$, two canonical entries indexed by the correction tag $\chi\in\{0,1\}$:
\[
\texttt{Memo}[\tilde{\sigma}] \;\in\; \big(\{\bot\}\cup\{0,1\}^n\big)^2.
\]
Each bucket stores the lexicographically minimal witness for its $(\tilde{\sigma},\chi)$ key. The buckets are disjoint: a new prefix competes only within its $\chi$-bucket for lexicographic minimality, so all semantically distinct witnesses are retained.

\paragraph{Aliased Canonical Normal Form (CNF).}
Fix an alias pair $(x_0,x_1)$ on indices $(i_0,i_1)$ in each half. Define a normalization map
\[
\mathsf{canon}(b) :=
\begin{cases}
b & \text{if } b[i_0]+b[i_1]\in\{0,2\}\ \text{or}\ (b[i_0],b[i_1])=(1,0),\\
b' & \text{if } (b[i_0],b[i_1])=(0,1),
\end{cases}
\]
where $b'$ is $b$ with the pair $(i_0,i_1)$ flipped from $(0,1)$ to $(1,0)$. The algorithm explores only states with $b=\mathsf{canon}(b)$ (canonical states); non-canonical states may be stored as witnesses but are flagged non-expandable. This quotients the search space by the alias-induced redundancy and is scheduler-agnostic (FIFO/LIFO/priority).

\paragraph{Expansion Policy (CNF, inlined).}
Equivalently, write $c(b):=b[i_0]+b[i_1]\in\{0,1,2\}$ and $\chi(b):=b[i_1]\in\{0,1\}$. A newly discovered state with mask $b$ is enqueued for extension iff
\[
(c(b)=0)\ \lor\ (c(b)=2)\ \lor\ \big(c(b)=1\ \wedge\ \chi(b)=0\big).
\]
States with $(c,\chi)=(1,1)$ are recorded in the appropriate bucket but \emph{not} expanded. This explicit policy yields the $3/4$ expansion accounting and is target-independent.

\paragraph{Inline Key Invariants.}
We use the following invariants in the main bound; proofs and extended details appear in Appendix~\ref{appendix:alias}.
\begin{itemize}
  \item \textbf{Lane Invariant.} For each $(\tilde{\sigma},\chi)$ the memo table stores the lex-minimal witness among all subsets mapping to that key; lanes $\chi\in\{0,1\}$ are disjoint.
  \item \textbf{Suppression-on-overwrite (lane-wise).} When a witness for $(\tilde{\sigma},\chi)$ is overwritten by a lexicographically smaller mask, the displaced state is never expanded.
  \item \textbf{CNF Safety.} Any non-canonical $(0,1)$ state has the same aliased key as its canonical $(1,0)$ counterpart and all of its descendants are \emph{shadowed} by the canonical branch; forbidding expansion from $(0,1)$ is therefore complete and sound.
\end{itemize}

\noindent
These invariants, particularly CNF safety, allow us to prove a strict reduction in
the number of states that must be explored. As we formally prove in
Lemma~\ref{lem:expansion-accounting} (Appendix~E), the CNF expansion policy ensures that
for any base set of choices outside the alias pair, at most three of the four
possible inclusion patterns are ever expanded. This deterministically prunes the
search space, guaranteeing that the number of expanded states in a half of size
$k$ is at most $\tfrac{3}{4} \cdot 2^k$.

\paragraph{Compensatory Merge.}
During the final merge phase we evaluate all valid interpretations of aliased sums. Let $\tilde{s}_L \in \tilde{\Sigma}(\ell_0)$ and $\tilde{s}_R \in \tilde{\Sigma}(\ell_1)$ be aliased sums, with alias pairs $(x_0,x_1)$ in $\ell_0$ and $(y_0,y_1)$ in $\ell_1$. For tags $\chi_L,\chi_R\in\{0,1\}$ define the corrected sums
\[
s_L(\chi_L)=\tilde{s}_L+\chi_L\,(x_1-x_0),\qquad
s_R(\chi_R)=\tilde{s}_R+\chi_R\,(y_1-y_0).
\]
We check the four combinations
\[
s_L(0)+s_R(0)=t,\quad s_L(1)+s_R(0)=t,\quad s_L(0)+s_R(1)=t,\quad s_L(1)+s_R(1)=t.
\]
This ensures correctness of the decision logic under value collision.

\medskip
\paragraph{Correctness Guarantee.}
The Controlled Aliasing mechanism is guaranteed to find a solution if one exists, without producing false positives. We formalize this below.

\begin{theorem}[Correctness of Controlled Aliasing]
\label{thm:correctness-controlled-aliasing}
The Controlled Aliasing solver, using the 2-lane $(\tilde{\sigma},\chi)$ memoization structure and the compensatory merge logic above, correctly solves the decision variant of \textsc{Subset-Sum}. Furthermore, if a solution exists, the witness reconstruction procedure returns a valid subset $W \subseteq S$.
\end{theorem}

\begin{proof}
The proof proceeds via the lane invariant, lane-wise suppression-on-overwrite, and CNF safety (inline invariants above), plus the soundness/completeness of the $2\times 2$ corrected merge (Appendix~\ref{appendix:alias}). Together these imply the theorem.
\end{proof}

\paragraph{Worst-Case Complexity.}
We now state the worst-case running-time guarantee obtained from Controlled Aliasing.

\begin{theorem}[Worst-Case Complexity of Controlled Aliasing under CNF]
\label{thm:aliasing-complexity}
Consider the interleaved double–meet-in-the-middle solver that applies a single alias pair $(x_0,x_1)$ with $x_0\ne x_1$ in each half $\ell_0,\ell_1$, uses the 2-lane $(\tilde{\sigma},\chi)$ memoization structure, performs compensatory merging as in §\ref{sec:aliasing-speedup}, and explores only canonical aliased states $b=\mathsf{canon}(b)$ (CNF; Appendix~\ref{appendix:alias}). Then the deterministic running time satisfies
\[
T(n)\;=\;O^*\!\left(2^{\,n/2-\varepsilon}\right)
\quad\text{with}\quad
\varepsilon\;=\;\log_2\!\left(\tfrac{4}{3}\right)\approx 0.415,
\]
and the space usage is $O^*(2^{n/2})$. The randomized variant achieves the same $O^*\!\left(2^{\,n/2-\varepsilon}\right)$ bound in expectation.
\end{theorem}

\emph{(Scope and unconditional bound.)}
If a half $\ell_i$ has fewer than two distinct element values, CNF is vacuous on that side and we
revert to the certificate-sensitive bound of Theorem~\ref{thm:cert-sensitive}:
$T(n)=O\!\big((U_0+U_1)\,n^2\big)$ deterministically and $O\!\big((U_0+U_1)\,n\big)$ in expectation,
which is $O^*(2^{n/2})$ in the worst case. Otherwise each half has at least two distinct values,
so an alias pair exists per half and the bound above applies, yielding
$O^*\!\left(2^{\,n/2-\log_2(4/3)}\right)$.

\begin{proof}
Immediate from Lemma~\ref{lem:expansion-accounting} (each half expands at most
$E_i \le \tfrac{3}{4}\,2^{k}$ with $k=n/2$) and
Theorem~\ref{thm:cert-sensitive} applied with $E_i$ in place of $U_i$.
\end{proof}

\begin{remark}[Conservative statements]
We state slightly conservative bounds for readability; the following routine tightenings
are available without changing the algorithm:
(i) Space under aliasing is $O((E_0+E_1)\,n)=O^*(2^{\,n/2-\log_2(4/3)})$;
(ii) The randomized bound holds w.h.p.\ using standard load-controlled hashing;
(iii) A deterministic time $O(U\,n\log n)$ follows by replacing hashing with balanced maps
and using $O(\log n)$ lex-compare primitives.
We omit these proofs as they are straightforward variations on the analyses given.
\end{remark}

\section{A Nuanced Structural View of Fine-Grained Hardness}
\label{sec:finegrained-hardness}

Our algorithm adds a structural dimension—via the number of distinct subset sums $U$ (equivalently, low collision entropy $H_c(S)=n-\log_2 U$)—and makes explicit that many classical hardness arguments implicitly target a structure-resistant slice of instances. We formalize this dependence.

\subsection{The Collision-Resistant Slice}
We begin by formalizing the slice of \SSP{} where classical worst-case barriers are most informative.

\begin{definition}[Collision-Resistant \SSP{}]
\label{def:collision-promise}
For $\delta\in(0,1]$, define
\[
\SSP^{\ge \delta} \;:=\; \{\, (S,t)\ :\ |\Sigma(S)| \ge 2^{\delta n}\, \}.
\]
A reduction to \SSP{} is \emph{collision-resistant} if, for some constant $c>0$, its outputs lie in $\SSP^{\ge c}$.
\end{definition}

\noindent\emph{Example.} The classical powers-of-two SAT$\to$\SSP{} encoding yields $U=2^n$, i.e., $(S,t)\in\SSP^{\ge 1}$.

\medskip
The Horowitz--Sahni algorithm runs in $O^*(2^{n/2})$ time. Thus, to preclude algorithms that beat this bound under a certificate-sensitive lens, a SETH-based hardness claim should (at minimum) target instances in $\SSP^{\ge 1/2}$, i.e., enforce $U \ge 2^{n/2}$ on the image of the reduction. This pinpoints the regime where adaptive methods cannot short-circuit via structural collisions and where classical worst-case lower bounds are most informative.

\subsection{Refining SETH Lower Bounds}
Our view reframes unconditional $2^{n/2}$ hardness claims for general \SSP{}: such claims are meaningful primarily on the structure-resistant slice $\SSP^{\ge \delta}$. We state this formally.

\begin{theorem}[SETH under entropy constraints]
\label{thm:fgc-revised}
Assume SETH. Then no $O^*(2^{\epsilon n})$ algorithm can solve all of $\SSP^{\ge \delta}$ for any $\epsilon<\delta$.
\end{theorem}

\begin{proof}[Proof sketch]
Take a collision-resistant reduction from $k$-SAT that maps instances to $(S,t)\in\SSP^{\ge \delta}$ with polynomial blowup. If an $O^*(2^{\epsilon n})$ algorithm existed for all such $(S,t)$ with $\epsilon<\delta$, composing would yield an $O^*(2^{\epsilon' m})$ algorithm for $k$-SAT with $\epsilon'<1$, contradicting SETH.
\end{proof}

\section{Proof-of-Concept Experiments}\label{sec:empirical-validation}

We provide \emph{proof-of-concept} experiments to sanity-check the theory. The goal is not a systems study but to verify that the observed running time tracks the structural parameter that our analysis identifies. We benchmark on synthetic instances where we can precisely control structure, using $n=48$ so that full per-half enumeration is feasible on a commodity desktop CPU. Throughout we split $S$ into two halves by alternating indices, so each half has size $k=n/2=24$. For each instance we fully enumerate \emph{both} halves to obtain $U_0:=|\Sigma(\ell_0)|$ and $U_1:=|\Sigma(\ell_1)|$, and we measure wall-clock time of the deterministic enumerator (hashing disabled) to avoid randomness. Our hypothesis is that introducing additive redundancy collapses $U_i$, and runtime scales proportionally with $U_0+U_1$ (hence with $\widehat U:=\max\{U_0,U_1\}$ for fixed $n$).

We manipulate three knobs known to induce subset-sum collisions: numeric density, duplicate elements, and additive progressions. For each setting we run the \emph{full} certificate construction on each half (no early termination) to measure the cost of generating $\Sigma(\ell_i)$. Additional methodology appear in Appendix~\ref{app:experiments}.

\begin{itemize}
\item \textbf{Numeric density.} Restricting values to a smaller bit-length $w$ forces collisions by the pigeonhole principle.
\item \textbf{Duplicate elements.} Introducing identical elements creates trivial additive dependencies.
\item \textbf{Additive progressions.} Planting short arithmetic progressions creates correlated dependencies.
\end{itemize}

\begin{table}[H]
\centering
\caption{Effect of structure on the per-half distinct-sum count. We report the ratio $U_i/2^{k}$ with $k=n/2=24$ (median over halves for a representative run).}
\label{tab:exp-summary}
\begin{tabular}{@{}lcc@{}}
\toprule
\textbf{Structural knob} & \textbf{Setting} & $\displaystyle U_i/2^{k}$ \\
\midrule
Numeric density ($n{=}48$) & $w = 32 \to 24 \to 16$ & $1.00 \to 0.84 \to 0.027$ \\
Duplicates ($n{=}48$) & $0 \to 2 \to 4$ duplicates & $1.00 \to 0.57 \to 0.32$ \\
Additive progressions ($n{=}48$) & 1 seq (len 4) / 2 seq (len 4) & $0.69$ / $0.47$ \\
\bottomrule
\end{tabular}
\end{table}

\paragraph{Results (sanity check).}
We observe large, systematic variation in $U_i$ under these perturbations: uniform-like inputs typically satisfy $U_i \approx 2^{k}$, whereas structured inputs collapse by orders of magnitude. For fixed $n$, the measured runtime of \icsubsetsum{} varies proportionally with $U_0+U_1$ (equivalently with $\widehat U$ up to a factor~2), consistent with the proven $O\!\big((U_0{+}U_1)\,n^2\big)$ bound. These experiments corroborate the central claim that \emph{structure (collisions) $\Rightarrow$ smaller $U$ $\Rightarrow$ faster enumeration}.

\section{Related Work}
\label{sec:related-work}

\paragraph{Subset Sum Algorithms.}
The classical Horowitz--Sahni meet-in-the-middle algorithm~\cite{horowitz1974subset} remains the standard worst-case baseline for \SSP{}, with $O^*(2^{n/2})$ runtime. A space-time refinement due to Schroeppel and Shamir achieves the same time with $O^*(2^{n/4})$ space~\cite{schroeppel-shamir-1981}. Despite extensive efforts, no \emph{deterministic} unconditional improvement below $2^{n/2}$ was known in general~\cite{bringmann2024subsetsumunconditional,openproblems_woeginger}. Our approach departs from this tradition in two ways: (i) we analyze runtime as a function of the number of \emph{distinct} subset sums $U$ rather than just $n$, yielding certificate-sensitive bounds; and (ii) via Controlled Aliasing under CNF with compensatory merge, we obtain a deterministic worst-case improvement.

\smallskip

A recent randomized algorithm by Bringmann, Fischer, and Nakos~\cite{bringmann2024subsetsumunconditional} gives the first unconditional improvement over Bellman’s dynamic programming for the target-constrained decision problem, running in $\widetilde{O}(|S(X,t)|\cdot \sqrt{n})$ time. Their technique focuses on reachable sums up to a target $t$ via algebraic/combinatorial methods. By contrast, our algorithm deterministically enumerates the full set of unique subset sums $\Sigma(S)$, supports canonical constructive witnesses, and adapts to instance structure. The two lines are complementary: they address different problem variants and performance measures. We also note pseudo-polynomial improvements for dense/structured instances (e.g.,~\cite{koiliaris-xu-2017,bringmann2017near}) that are orthogonal to our certificate-size viewpoint.

\paragraph{Instance Complexity and Adaptive Algorithms.}
Instance Complexity was introduced by Orponen–Ko–Schöning–Watanabe~\cite{orponen1994instance} to capture the inherent difficulty of individual instances. Although potent as a conceptual tool, IC was long viewed as non-constructive (see, e.g., \cite{Fortnow2004KolmogorovCA}). Our results give a constructive realization: we generate the IC certificate online with runtime scaling in the certificate size $U$, and we wrap it to satisfy the partial-correctness conventions of IC (see §\ref{sec:ic-background}).

\paragraph{Fine-Grained Complexity.}
The fine-grained framework~\cite{williams2018some} provides tight conditional barriers (e.g., under SETH). We add a structural dimension: such lower bounds for \SSP{} implicitly target collision-resistant instances with near-maximal $U$. Making this dependency explicit complements work on compressibility-aware reductions~\cite{abboud2019bicriteria}; our formulation via $\Sigma(S)$ and collision entropy is instance-level and operational.

\paragraph{Collision Structure in Combinatorics.}
The role of collisions in subset sums appears in additive combinatorics and cryptanalysis. Austrin-Kaski-Koivisto-Nederlof~\cite{ss_in_the_absence_of_concentration} analyze structure under anti-concentration, while Howgrave-Graham and Joux~\cite{howgrave2010knapsack} exploit collisions via randomized meet-in-the-middle techniques to obtain sub-$2^{n/2}$ behavior in certain regimes. Classic Littlewood--Offord–type results~\cite{tao2009inverse} bound collision counts via anti-concentration. Our contribution is algorithmic and instance-sensitive: we prune duplicate sums in real time and link entropy collapse directly to runtime; and, distinct from prior randomized collision tricks, our Controlled Aliasing (under CNF with compensatory merge) is a deterministic, target-independent transformation that preserves correctness while provably shrinking the explored state space.

\section{Future Directions}
\label{sec:future}

\paragraph{Closing Algorithmic Gaps for Subset Sum.}
Two natural questions remain. First, for the randomized $O(U \cdot n)$ algorithm, can the linear factor in $n$ be removed to obtain $\tilde{O}(U)$ time? Second, and more fundamentally, can the deterministic $O(U \cdot n^2)$ bound be improved to match the randomized $O(U \cdot n)$, e.g., via a deterministic canonicality test that avoids the $O(n)$ lexicographic comparison cost while preserving the local model in Conjecture~\ref{conj:local-lb}?

\smallskip

\paragraph{Improving the Controlled Aliasing Bound.}
Our worst-case speedup $\varepsilon=\log_2(\frac{4}{3})$ arises from a single alias pair per half under the 2-lane canonicalization keyed by $(\tilde{\sigma},\chi)$ (CNF). Straightforward pattern-based generalizations for alias groups of size $g$ induce $2^g$ lanes and exponential overhead, keeping $\varepsilon$ constant. We conjecture that a \emph{count-based} scheme—tracking only the number of chosen elements from each alias group (yielding $(g{+}1)$ buckets) with polynomial overhead—could permit $g=\Theta(\log n)$, improving the multiplicative reduction from $2^g$ to $(g{+}1)$. In principle, this could yield $\varepsilon=\Theta(\log n)$, i.e., a worst-case runtime of $2^{n/2}/n^{c}$ for some constant $c>0$. Making this rigorous requires (i) a safety proof for count-only canonicalization (lane invariant and merge completeness), and (ii) a certified reconstruction path from counts to witnesses without reintroducing exponential blowup.

\smallskip

\paragraph{Structure-Sensitive Runtimes in NP.}
Which other NP-complete problems admit certificate-sensitive runtimes? Collision entropy provides a structural knob complementary to solution density. For Partition, Knapsack, and Bin Packing, one may aim to construct compressed summaries of feasible configurations whose \emph{construction time} is proportional to summary size. Formalizing such certificates and proving instance-sensitive enumeration bounds are promising directions.

\smallskip

\paragraph{Derandomization of Local Canonicality.}
The randomized $O(U \cdot n)$ bound uses 2-universal hashing to filter collisions so that lexicographic comparisons are needed only rarely (expected $O(1)$ per check). Is there a deterministic analogue achieving sublinear-time canonicality tests in the local deduplicating model? Such a result would match the conjectured optimal runtime (Conjecture~\ref{conj:local-lb}) and yield a practical local derandomization.

\smallskip

\paragraph{Constructive Instance Complexity Beyond Subset Sum.}
We showed that the certificate $\Sigma(S)$ can be constructed in deterministic $O(U \cdot n^2)$ time (or expected $O(U \cdot n)$). Can analogous online certifiers be built for problems such as 3SAT or Clique under appropriate entropy/compressibility assumptions? A broader theory of \emph{algorithmic IC}—tracking both information content and construction cost—remains to be developed.

\smallskip

\paragraph{Toward a Complexity Class for Certificate-Sensitive Problems.}
Our solver runs in time polynomial in the input size $n$ and in the size of an intrinsic instance-specific certificate ($U=|\Sigma(S)|$). This suggests a class we tentatively call \textbf{P-IC} (\emph{Polynomial in Instance Certificate}): problems for which a natural instance certificate can be constructed and then used to solve the instance in time polynomial in the certificate size and input size. Pinning down robust certificate notions (e.g., sets of satisfying assignments for SAT, achievable value/weight pairs for Knapsack) and charting the boundaries of P-IC is an intriguing direction for future work.

\section{Conclusion}
\label{sec:conclusion}

\paragraph{Theoretical Advance.}
This paper develops a framework for structure-aware enumeration in NP-complete problems, using \textsc{Subset Sum} as a case study. We introduce the certificate-sensitive parameter $U:=|\Sigma(S)|$ and show that it aligns with both per-instance difficulty and observed runtime. The results help explain why many real-world instances are easier, why SETH-based reductions can be structurally brittle, and how instance complexity can be made constructive.

\paragraph{Algorithmic Advance.}
\icsubsetsum{} delivers a set of guarantees that advance the state of the art:
\begin{itemize}
\paragraph{Certificate-sensitive runtime.}

  \item \textbf{Certificate-sensitive runtime.} A deterministic worst-case bound of $O(U\cdot n^2)$ and an
    expected randomized bound of $O(U\cdot n)$, realizing a constructive link to Instance Complexity by tying performance to the certificate size $U$, which is discovered \emph{online} in an \emph{anytime} fashion.
  \item \textbf{Worst-case improvement.} Via Controlled Aliasing under count-based canonical normalization (CNF) at the alias indices, a deterministic bound of $O^*(2^{\,n/2-\varepsilon})$ with $\varepsilon=\log_2(\frac{4}{3})\approx 0.415$, improving on the classical $O^*(2^{n/2})$ Horowitz--Sahni bound for all sufficiently large $n$.
  \item \textbf{Real-time certificate generation.} Online construction of the constructive certificate, storing the lex-minimal witness for each $\sigma \in \Sigma(S)$ without duplication.
\end{itemize}
Closing the gap between the deterministic and randomized bounds remains a key open challenge, as does the question of whether an $\tilde{O}(U)$ runtime is achievable.\footnote{May all your programs be short and quick (Fortnow \cite{Fortnow2004KolmogorovCA})
and may their runtimes reveal the structure within.}

\bibliographystyle{alpha}
\bibliography{icsusbsetsum}

\appendix

\section*{Appendix Overview}

This appendix expands on the enumeration logic, correctness conditions, and structural optimizations introduced in Sections~\ref{sec:certificate-framework}--\ref{sec:finegrained-hardness}. It is organized as follows:
 
\begin{itemize}
    \item Appendix~\ref{app:column-enum} details the column-wise subset-sum enumerator, including the \emph{suppression-on-overwrite} invariant and the lex-minimal canonicalization rule, with full pseudocode.
     \item Appendix~\ref{app:dmitm} presents the interleaved double–meet-in-the-middle framework, the real-time \textsc{Check}/\textsc{CheckAliased} procedures (including the $2\times 2$ compensatory merge under aliasing), and the work accounting used in the certificate-sensitive bounds.
   
    \item Appendix~\ref{app:combine} states the algebraic combination lemmas used by the solver (direct, complement, and mixed cases) and the associated correctness checks for sum recombination that power the real-time \textsc{Check} routines.
    \item Appendix~\ref{app:experiments} provides additional experimental data, including plots of runtime scaling versus $U$ and measured collision entropy.
    \item Appendix~\ref{appendix:alias} gives the full \emph{Controlled Aliasing} framework: the 2-lane $(\tilde{\sigma},\chi)$ memoization with affine correction, the count-based canonical normalization (CNF) at the alias indices, and the expansion policy. It includes the key lemmas—Lemma~\ref{lem:alias-reduction}, Lemma~\ref{lem:cnf-shadow}, Lemma~\ref{lem:cnf-preserve}, Lemma~\ref{lem:expansion-accounting}, Lemma~\ref{lem:policy-safety}—and the proof of Theorem~\ref{thm:aliasing-complexity}.
\end{itemize}

\section{Column-wise Enumerator: Full Details}
\label{app:column-enum}

\subsection*{Algorithm Description}

\paragraph{Column-wise expansion.}
We organize the search by \emph{columns} of increasing subset size. The state space consists of canonical partial sums $\sigma$ paired with their current prefix (bitmask) $R$. At column $k$ we extend each canonical state of size $k{-}1$ by adding one unused element $a_i \notin R$, producing $\sigma'=\sigma+a_i$ with prefix $R' = R \cup \{a_i\}$. If $\sigma'$ is new, or if $R'$ is lexicographically smaller than the current representative, we update the representative of $\sigma'$. This yields at most $O(U\cdot n)$ total \emph{extension attempts} from canonical states. As established in Section~\ref{sec:subset-sum-enumerator}, the overall complexity is governed by the cost of canonicality checks during collisions.

\paragraph{High-density handling.}
To guarantee correctness in high-collision regimes, we impose a deterministic canonicalization based on lexicographic order of index bitmasks: scan indices $1\!\to\! n$ and prefer $0<1$ at the first differing position. Each numeric sum $\sigma$ keeps exactly one \emph{canonical} representative (lex-minimal mask), which prevents non-deterministic pruning and ensures a single expansion path per sum.

\paragraph{4-Column Litmus Test.}
\label{app:litmus-test}
The real-time collision checks enable a quick hardness probe. By running only the first $k$ columns (e.g., $k=4$) in \textbf{polynomial} time $O(n^k)$, we empirically estimate the early collision rate. High early collision typically predicts small $U$, signaling an “easy’’ instance; low early collision predicts a large $U$, signaling a harder, unstructured input. This offers a lightweight pre-check before any full exponential exploration.

\subsection*{Reader’s Guide (what to look for).}
\begin{itemize}
  \item \emph{Separation of concerns:} (i) enumeration of prefixes, (ii) canonical tie-breaking, (iii) suppression-on-overwrite. Any scheduler (BFS/DFS/priority) that respects these yields the same canonical set.
  \item \emph{Spec vs.\ implementation:} We first give a \emph{minimal recursive} specification that makes the logic transparent; the column-wise iterative code that follows is the optimized implementation used in our analysis.
  \item \emph{Where the $O(U\!\cdot\!n^2)$ comes from:} $O(U\!\cdot\!n)$ extension attempts times $O(n)$ lex-compare cost on collisions, with suppression ensuring each canonical state expands at most once.
\end{itemize}

\subsection*{Minimal Recursive Specification (for clarity)}
\begin{algorithm}[H]
\caption{\textsc{EnumerateUniqueSubsetSums-Recursive} (spec; with suppression-on-overwrite)}
\label{alg:enum-recursive-spec}
\begin{algorithmic}[1]
\Require $S=\{a_1,\dots,a_n\}$ with fixed index order
\Ensure \texttt{Memo} maps each numeric sum to its canonical (lex-min) bitmask
\State \texttt{Memo} $\gets \{\,0 \mapsto \texttt{empty\_bitmask}\,\}$;\quad \texttt{doNotExtend}[$\texttt{empty\_bitmask}$] $\gets$ \texttt{false}
\Procedure{Recurse}{$\texttt{mask},\,\texttt{sum},\,\texttt{next}$} \Comment{include-only-forward indices ensure a unique generation path}
  \If{\texttt{doNotExtend}[\texttt{mask}]} \Return \EndIf
  \For{$i=\texttt{next}$ \textbf{to} $n$}
    \State $\texttt{new\_sum}\gets \texttt{sum}+a_i$;\quad $\texttt{new\_mask}\gets \Call{SetBit}{\texttt{mask},i}$
    \If{$\texttt{new\_sum}\notin\texttt{Memo}$}
      \State \texttt{Memo[$\texttt{new\_sum}$]} $\gets \texttt{new\_mask}$;\quad \texttt{doNotExtend[$\texttt{new\_mask}$]} $\gets$ \texttt{false}
      \State \Call{Recurse}{$\texttt{new\_mask},\,\texttt{new\_sum},\,i{+}1$}
    \ElsIf{\Call{IsLexicographicallySmaller}{$\texttt{new\_mask}$, \texttt{Memo[$\texttt{new\_sum}$]}}}
      \State $\texttt{old}\gets \texttt{Memo}[\texttt{new\_sum}]$;\quad \texttt{doNotExtend[$\texttt{old}$]} $\gets$ \texttt{true} \Comment{suppress-on-overwrite}
      \State \texttt{Memo[$\texttt{new\_sum}$]} $\gets \texttt{new\_mask}$;\quad \texttt{doNotExtend[$\texttt{new\_mask}$]} $\gets$ \texttt{false}
      \State \Call{Recurse}{$\texttt{new\_mask},\,\texttt{new\_sum},\,i{+}1$}
    \EndIf
  \EndFor
\EndProcedure
\State \Call{Recurse}{$\texttt{empty\_bitmask},\,0,\,1$};\quad \Return \texttt{Memo}
\end{algorithmic}
\end{algorithm}

\subsection*{Mechanics of \texttt{Frontier}/\texttt{Rep}/\texttt{doNotExtend} (matches Inv.~\ref{inv:suppress})}
\begin{itemize}
  \item \textbf{\texttt{Frontier}} holds exactly the canonical states scheduled to be \emph{expanded next}. We advance by columns (subset size), so each state in \texttt{Frontier} is expanded at most once per Inv.~\ref{inv:suppress}.
  \item \textbf{\texttt{Rep}[\,$\sigma$\,]} points to the \emph{current representative state} for numeric sum $\sigma$ in this column’s scheduler. It lets us (i) find and (ii) immediately suppress a displaced representative when a lex-smaller mask appears.
  \item \textbf{\texttt{doNotExtend}} is a per-state flag. When a representative is overwritten by a lex-smaller mask, we set \texttt{doNotExtend} on the displaced state. Even if a displaced state were still present in some queue, the flag enforces \emph{no further expansion}, ensuring each canonical sum expands exactly once.
  \item \textbf{Optional unscheduling.} We also remove the displaced state from the current \texttt{NextFrontier} if it was tentatively enqueued, which avoids a wasted dequeue in this same column. This is a micro-optimization; correctness relies only on \texttt{doNotExtend}.
\end{itemize}

\subsection{Iterative Column-wise Implementation (optimized)}.
\begin{algorithm}[H]
\caption{\textsc{EnumerateUniqueSubsetSums} (with suppression-on-overwrite)}
\label{alg:enum-revised}
\begin{algorithmic}[1]
\Require Set $S=\{a_1,\dots,a_n\}$ (indices $1..n$)
\Ensure \texttt{Memo} maps each numeric sum to its canonical (lex-min) bitmask

\State $\texttt{Memo} \gets \{\,0 \mapsto \texttt{empty\_bitmask}\,\}$ \Comment{initial canonical rep for sum $0$}
\State $\texttt{Frontier} \gets \{\;\texttt{State}(0,\texttt{empty\_bitmask},\texttt{doNotExtend}= \texttt{false})\;\}$
\State $\texttt{Rep} \gets \{\,0 \mapsto \text{the element of }\texttt{Frontier}\,\}$ 
\Comment{\texttt{Rep} tracks the \emph{current} representative state per sum}

\For{$k = 1$ \textbf{to} $n$} \Comment{BFS by columns (subset size); each canonical state expands at most once}
  \State $\texttt{NextFrontier} \gets \emptyset$
  \ForAll{$\texttt{state} \in \texttt{Frontier}$}
    \If{\texttt{state.doNotExtend}} \textbf{continue} \EndIf \Comment{Inv.~\ref{inv:suppress}: suppressed reps never expand}
    \For{$i = 1$ \textbf{to} $n$}
      \If{\textbf{bit} $i$ \textbf{is not set in} \texttt{state.mask}}
        \State $\texttt{new\_sum} \gets \texttt{state.sum} + a_i$
        \State $\texttt{new\_mask} \gets \Call{SetBit}{\texttt{state.mask}, i}$
        \If{$\texttt{new\_sum} \notin \texttt{Memo}$}
          \Comment{first time we see this numeric sum: install canonical rep and schedule it}
          \State $\texttt{Memo[new\_sum]} \gets \texttt{new\_mask}$
          \State $\texttt{new\_state} \gets \texttt{State}(\texttt{new\_sum},\texttt{new\_mask},\texttt{false})$
          \State $\texttt{Rep[new\_sum]} \gets \texttt{new\_state}$
          \State $\texttt{NextFrontier} \gets \texttt{NextFrontier} \cup \{\texttt{new\_state}\}$
        \ElsIf{\Call{IsLexicographicallySmaller}{\texttt{new\_mask}, \texttt{Memo[new\_sum]}}}
          \Comment{lex-smaller: overwrite the rep and suppress the displaced one (Inv.~\ref{inv:suppress})}
          \If{$\texttt{Rep[new\_sum]} \neq \texttt{nil}$}
            \State $\texttt{Rep[new\_sum].doNotExtend} \gets \texttt{true}$ 
            \Comment{suppress old representative from any future expansion}
            \State \textit{/* optional micro-optimization: if scheduled, unschedule it in this column */}
            \State remove $\texttt{Rep[new\_sum]}$ from $\texttt{NextFrontier}$ if present
          \EndIf
          \State $\texttt{Memo[new\_sum]} \gets \texttt{new\_mask}$ 
          \Comment{install the new canonical representative}
          \State $\texttt{new\_state} \gets \texttt{State}(\texttt{new\_sum},\texttt{new\_mask},\texttt{false})$
          \State $\texttt{Rep[new\_sum]} \gets \texttt{new\_state}$ 
          \Comment{\texttt{Rep} now points at the (unique) expandable canonical state}
          \State $\texttt{NextFrontier} \gets \texttt{NextFrontier} \cup \{\texttt{new\_state}\}$
        \Else
          \Comment{non-canonical duplicate: no scheduling; current rep remains unique expandable state}
        \EndIf
      \EndIf
    \EndFor
  \EndFor
  \State $\texttt{Frontier} \gets \texttt{NextFrontier}$ 
  \Comment{advance one column; each canonical sum’s representative is expanded at most once}
  \If{$\texttt{Frontier} = \emptyset$} \textbf{break} \EndIf
\EndFor
\State \Return \texttt{Memo}
\end{algorithmic}
\end{algorithm}

\begin{algorithm}[H]
\caption{\textsc{IsLexicographicallySmaller} (1-indexed, $0<1$)}
\label{alg:lex-smaller}
\begin{algorithmic}[1]
\Function{IsLexicographicallySmaller}{$A$, $B$}
    \Comment{$A$, $B$ are $n$-bit masks}
    \For{$i = 1$ \textbf{to} $n$}
        \State $a_i \gets$ bit $i$ of $A$;\quad $b_i \gets$ bit $i$ of $B$
        \If{$a_i \neq b_i$} \Return $(a_i < b_i)$ \EndIf \Comment{prefer $0$ at first difference}
    \EndFor
    \State \Return \textbf{false} \Comment{equal masks}
\EndFunction
\end{algorithmic}
\end{algorithm}

\begin{algorithm}[H]
\caption{\textsc{SetBit} (helper)}
\label{alg:setbit}
\begin{algorithmic}[1]
\Function{SetBit}{$M, i$}
  \State \Return $M$ with bit $i$ set to $1$
\EndFunction
\end{algorithmic}
\end{algorithm}

\subsection{Implementation Invariants}\label{app:implementation-invariants}

To guarantee correctness and performance, the enumerator maintains:
\begin{itemize}
  \item \textbf{Deterministic tie-breaking.} For any numeric sum $\sigma$, the lex-minimal bitmask (Algorithm~\ref{alg:lex-smaller}) is the unique representative. This ensures global canonicity across columns.
  \item \textbf{Efficient storage.} We store at most $U$ representatives, each as an $n$-bit mask, for total space $O(U\cdot n)$.
  \item \textbf{Duplication elimination.} A numeric sum may be \emph{discovered} multiple times, but its canonical representative is \emph{expanded} at most once. All non-canonical discoveries either lose the tie or overwrite and suppress the previous representative.
  \item \textbf{Suppression on overwrite.} When a representative for $\sigma$ is overwritten by a lex-smaller mask, the displaced state is marked \emph{non-expandable} (\texttt{doNotExtend}$\gets$\texttt{true}) and, if queued, is optionally unscheduled. Consequently, each canonical sum is expanded exactly once over the entire enumeration.
\end{itemize}

\section{Double-MIM Solver: Detailed Pseudocode and Analysis}
\label{app:dmitm}

\subsection{Algorithm Sketch}

We split $S=\{a_1,\ldots,a_n\}$ into two halves $\ell_0$ and $\ell_1$ (sizes within $\pm1$) and run the unique-subset-sum enumerator on both sides \emph{interleaved}. Each time a new canonical state is discovered on one half, we immediately invoke a constant-time family of membership tests on the other half (the \textsc{Check} or \textsc{CheckAliased} routines from §\ref{app:combine}) to detect solutions \emph{online}. This avoids materializing $\Sigma(S)$ in full.

When \emph{Controlled Aliasing} is enabled (Section~\ref{sec:aliasing-speedup}; Appendix~\ref{appendix:alias}), each half uses 2-lane memoization keyed by $(\tilde{\sigma},\chi)$ and enforces the count-based canonical normal form (CNF): only masks $b$ with
\[
c(b):=b[i_0]+b[i_1]\in\{0,2\}\quad\text{or}\quad\big(c(b)=1\ \wedge\ b[i_1]=0\big)
\]
are enqueued for expansion; non-canonical states are stored but marked \texttt{doNotExtend}. This yields the $\tfrac{3}{4}$ expansion factor used in the worst-case bound.

\subsection{Core Pseudocode (Modular)}

We write \texttt{State(sum, mask, split, doNotExtend)} for a per-half enumeration state.
Lexicographic comparisons of bitmasks use Algorithm~\ref{alg:lex-smaller} (Appendix~\ref{app:column-enum}).
The solution checks are given in Algorithms~\ref{alg:check} and \ref{alg:check-aliased} (Appendix~\ref{app:combine}).

\begin{algorithm}[H]
\caption{\textsc{DoubleMIM-IC-SubsetSum} (interleaved, alias-aware)}
\label{alg:dmitm-main}
\begin{algorithmic}[1]
\Require Multiset $S$; target $t$; flag \texttt{AliasingEnabled};
\Statex\hspace{1.1em} alias pairs $(x_0^{(0)},x_1^{(0)})$ at indices $(i_0^{(0)},i_1^{(0)})$ on $\ell_0$, and $(x_0^{(1)},x_1^{(1)})$ at $(i_0^{(1)},i_1^{(1)})$ on $\ell_1$
\Ensure Returns a witness if $t\in\Sigma(S)$; else reports failure

\State Split $S$ into $\ell_0,\ell_1$ (e.g., alternating indices)
\State Initialize per-half tables: \ \texttt{Memo}$_b \gets \emptyset$, \ \texttt{Rep}$_b \gets \emptyset$ for $b\in\{0,1\}$
\State Initialize per-half frontiers with the empty state:
\Statex \hspace{1.6em} \texttt{Frontier} $\gets \{\texttt{State}(0,\mathbf{0},0,\texttt{false}),\ \texttt{State}(0,\mathbf{0},1,\texttt{false})\}$
\State Insert the empty key on both halves:
\Statex \hspace{1.6em} \texttt{InsertOrImprove}$(0,\mathbf{0},0)$; \ \texttt{InsertOrImprove}$(0,\mathbf{0},1)$
\For{$r=0$ \textbf{to} $\lfloor n/2 \rfloor$} \Comment{interleave by ``column''/Hamming weight}
  \State \texttt{NextFrontier} $\gets \emptyset$
  \ForAll{$p \in$ \texttt{Frontier}}
    \If{$p.\texttt{doNotExtend}$} \textbf{continue} \EndIf
    \State $b \gets p.\texttt{split}$;\ \ $L \gets \ell_b$;\ \ $k \gets |L|$
    \For{$i=0$ \textbf{to} $k-1$}
      \If{$p.\texttt{mask}[i]=0$}
        \State $s' \gets p.\texttt{sum} + L[i]$;\ \ $e' \gets p.\texttt{mask}$ with bit $i\leftarrow 1$
        \State \texttt{ProcessExtension}$(b, s', e')$ \Comment{Algorithm~\ref{alg:process-ext}}
      \EndIf
    \EndFor
  \EndFor
  \State \texttt{Frontier} $\gets$ \texttt{NextFrontier}
  \If{\texttt{Frontier} $=$ $\emptyset$} \textbf{break} \EndIf
\EndFor
\State \Return ``No solution found''
\end{algorithmic}
\end{algorithm}

\begin{algorithm}[H]
\caption{\textsc{ProcessExtension}$(b, s', e')$}
\label{alg:process-ext}
\begin{algorithmic}[1]
\Require Half index $b\in\{0,1\}$; candidate sum $s'$ and mask $e'$
\State /* Compute per-half key (aliased or not) */
\State $\texttt{key} \gets \textsc{KeyOf}(b, s', e')$ \Comment{Algorithm~\ref{alg:keyof}}
\State $wasNew \gets$ \texttt{InsertOrImprove}$(\texttt{key}, e', b)$ \Comment{Algorithm~\ref{alg:insert-improve}}
\If{$wasNew$}
  \State $q \gets \texttt{State}(s', e', b, \texttt{false})$
  \If{\texttt{AliasingEnabled} \textbf{and} \textbf{not} \textsc{Expandable}$(e', i_0^{(b)}, i_1^{(b)})$} \Comment{CNF gate}
    \State $q.\texttt{doNotExtend} \gets \texttt{true}$
  \EndIf
  \State \texttt{NextFrontier} $\gets$ \texttt{NextFrontier} $\cup \{q\}$
  \State \textbf{if} \texttt{AliasingEnabled} \textbf{then} \textsc{CheckAliased}$(q)$ \textbf{else} \textsc{Check}$(q)$
\EndIf
\end{algorithmic}
\end{algorithm}

\begin{algorithm}[H]
\caption{\textsc{KeyOf}$(b, s', e')$}
\label{alg:keyof}
\begin{algorithmic}[1]
\Require Half index $b$; real sum $s'$; mask $e'$
\If{\texttt{AliasingEnabled}}
  \State $\chi \gets e'[\,i_1^{(b)}\,]$
  \State $\tilde{s} \gets s' - \chi\cdot\big(x_1^{(b)}-x_0^{(b)}\big)$
  \State \Return $(\tilde{s}, \chi)$
\Else
  \State \Return $s'$
\EndIf
\end{algorithmic}
\end{algorithm}

\begin{algorithm}[H]
\caption{\textsc{InsertOrImprove}$(\texttt{key}, e', b)$ (suppression-on-overwrite)}
\label{alg:insert-improve}
\begin{algorithmic}[1]
\Require Key = $s'$ (no alias) or $(\tilde{s},\chi)$ (alias); mask $e'$; half $b$
\State $\texttt{Memo} \gets$ \texttt{Memo}$_b$;\quad \texttt{Rep} $\gets$ \texttt{Rep}$_b$
\If{$\texttt{key} \notin \texttt{Memo}$}
  \State \texttt{Memo}$[\texttt{key}] \gets e'$
  \State \texttt{Rep}$[\texttt{key}] \gets \texttt{State}(\ \cdot\ )$ \Comment{bound representative to new state lazily}
  \State \Return \textbf{true} \Comment{new canonical representative}
\ElsIf{\Call{IsLexicographicallySmaller}{$e'$, \texttt{Memo}[$\texttt{key}$]}}
  \State /* suppress old representative so it never extends further */
  \If{$\texttt{Rep}[\texttt{key}] \neq \texttt{nil}$} \ \ $\texttt{Rep}[\texttt{key}].\texttt{doNotExtend} \gets \texttt{true}$ \ \ \EndIf
  \State \texttt{Memo}$[\texttt{key}] \gets e'$
  \State \texttt{Rep}$[\texttt{key}] \gets \texttt{State}(\ \cdot\ )$
  \State \Return \textbf{true} \Comment{canonical representative updated}
\Else
  \State \Return \textbf{false}
\EndIf
\end{algorithmic}
\end{algorithm}

\paragraph{Notes on Modularity.}
\begin{itemize}
  \item \textbf{Enumeration Logic.} \textsc{DoubleMIM-IC-SubsetSum} drives interleaved enumeration; \textsc{ProcessExtension} encapsulates “generate, deduplicate, (optionally) gate by CNF, then check”.
  \item \textbf{Deduplication.} \textsc{InsertOrImprove} implements the lex-minimal canonicalization with \emph{suppression-on-overwrite}, ensuring each canonical key is expanded at most once.
  \item \textbf{Aliasing Hooks.} \textsc{KeyOf} provides a single switch for aliased vs.\ true sums; the CNF gate is the call to \textsc{Expandable}$(\cdot)$ (Appendix~\ref{appendix:alias}).
  \item \textbf{Solution Checks.} The real-time checks defer to Algorithms~\ref{alg:check} and~\ref{alg:check-aliased}.
\end{itemize}

\subsection{Time and Space Bounds}
\label{app:dmitm-bounds}

Let $U_i:=|\Sigma(\ell_i)|$ and $\widehat U:=\max(U_0,U_1)$.

\paragraph{Without Aliasing.}
By Theorems~\ref{thm:enum-det-runtime} and~\ref{thm:enum-rand-runtime}, enumerating each half in isolation costs $O(U_i\cdot (n/2)^2)$ deterministically, and $O(U_i\cdot (n/2))$ in expectation for the randomized variant. In the interleaved solver, every cross-half probe and canonical comparison can be \emph{injectively charged} to a unique successor attempt on the dominant side (the half with $U_i=\widehat U$) due to suppression-on-overwrite. Hence
\[
T_{\mathrm{det}}(n)=O\!\big(\widehat U\cdot n^2\big),\qquad
\mathbb{E}\big[T_{\mathrm{rand}}(n)\big]=O\!\big(\widehat U\cdot n\big),
\]
which is asymptotically equivalent to the sum form $O\!\big((U_0{+}U_1)\cdot n^2\big)$ and $O\!\big((U_0{+}U_1)\cdot n\big)$ stated in Theorem~\ref{thm:cert-sensitive}.
Space is $O\!\big((U_0{+}U_1)\cdot n\big)$ for storing canonical bitmasks.

\paragraph{With Controlled Aliasing (CNF).}
Let $E_i$ be the number of \emph{expanded} states on half $i$ under CNF. By Lemma~\ref{lem:expansion-accounting}, $E_i \le \tfrac{3}{4}\cdot 2^{|\ell_i|} = \tfrac{3}{4}\cdot 2^{n/2}$ in the worst case. Substituting $E_i$ for $U_i$ in the same accounting yields
\[
T_{\mathrm{det}}(n)=O\!\big(\max(E_0,E_1)\cdot n^2\big)=O^*\!\left(2^{\,n/2-\log_2(4/3)}\right), 
\]

\[
\qquad
\mathbb{E}\big[T_{\mathrm{rand}}(n)\big]=O\!\big(\max(E_0,E_1)\cdot n\big)=O^*\!\left(2^{\,n/2-\log_2(4/3)}\right),
\]

matching Theorem~\ref{thm:aliasing-complexity}. Space remains $O^*(2^{n/2})$, since each aliased key stores up to two canonical masks (one per lane $\chi\in\{0,1\}$).

\paragraph{Optimality (Conditional).}
Under Conjecture~\ref{conj:local-lb}, the randomized solver is optimal up to constants in the local deduplicating model, achieving $\Theta(U\cdot n)$ expected time (Theorem~\ref{thm:conditional-opt}). The deterministic $O(U\cdot n^2)$ bound leaves a polynomial gap, whose removal would require a deterministic sublinear-time canonicality test per extension.

\medskip
\noindent\emph{Summary.} The interleaved double–meet-in-the-middle framework, combined with (i) lex-minimal canonicalization with suppression-on-overwrite and (ii) the CNF expansion gate under Controlled Aliasing, yields (a) certificate-sensitive runtimes scaling with $U_0,U_1$ and (b) a strict worst-case improvement below $2^{n/2}$ on \emph{all} inputs.

\section{Compositional Sum Lemmas and Complement Check}
\label{app:combine}

These lemmas govern how partial sums are combined across recursive calls and meet-in-the-middle partitions. They apply to both aliasing and non-aliasing regimes and are invoked by all canonical certificate routines to ensure uniqueness and correctness.

Let \(\ell_0,\ell_1\) be the two halves of a split of \(S\),
and write \(\Sigma(\ell_i)\) for the set of distinct subset sums
produced by the enumerator on side \(i\).

\begin{lemma}[Direct Match in One Half]
\label{lem:single-hit}
A subset \(A \subseteq \ell_i\) is a witness to \((S, t)\) if
\[
  \sigma(A) = t.
\]
\end{lemma}

\begin{lemma}[Self-Complement Match]
\label{lem:combine-single-comp}
A subset \(A \subseteq \ell_i\) solves \((S, t)\) if
\[
  \sigma(\ell_i \setminus A) = \sigma(\ell_i) - t.
\]
\end{lemma}

\begin{lemma}[Cross-Half Direct Match]
\label{lem:combine-direct}
For subsets \(A \subseteq \ell_0\) and \(B \subseteq \ell_1\),
\[
  \sigma(A) + \sigma(B) = t
  \quad\Longleftrightarrow\quad
  t - \sigma(A) \in \Sigma(\ell_1).
\]
\end{lemma}

\begin{lemma}[Double Complement Match]
\label{lem:combine-double-comp}
If subsets \(A \subseteq \ell_0\) and \(B \subseteq \ell_1\) satisfy
\[
  \sigma(\ell_0 \setminus A) + \sigma(\ell_1 \setminus B) = \sigma(S) - t,
\]
then \(S \setminus (A \cup B)\) is a witness (i.e., sums to \(t\)).
Equivalently,
\[
  \sigma(\ell_0 \setminus A) + \sigma(\ell_1 \setminus B) = \sigma(S) - t
  \ \Longleftrightarrow\ 
  \sigma(A) + \sigma(B) = t.
\]
\end{lemma}

\begin{lemma}[Mixed Case Match]
\label{lem:combine-mixed}
A subset \(A \subseteq \ell_0\) and the complement of a subset
\(B \subseteq \ell_1\) form a solution if and only if
\[
  \sigma(A) + \sigma(\ell_1 \setminus B) = t.
\]
\end{lemma}

\begin{proof}[Proof sketch]
All statements follow from elementary algebra on disjoint subsets and set complements (e.g., \(\sigma(\ell_i \setminus A)=\sigma(\ell_i)-\sigma(A)\) and \(\ell_0\) is disjoint from \(\ell_1\)). These identities justify the conditional checks performed by the \textsc{Check} routine in the main solver. Since the enumerator generates all elements of \(\Sigma(\ell_i)\) correctly (§\ref{app:column-enum}), all valid witnesses are discoverable by the algorithm.

Full derivations appear in §6.1–6.3 of \cite{bwss2025}.
\end{proof}

\medskip
\noindent\textbf{Real-Time Solution Check Procedures (invoked during enumeration).}

\begin{algorithm}[H]
\caption{\textsc{Check}$(q)$ — standard (no aliasing)}
\label{alg:check}
\begin{algorithmic}[1]
\Require $q$ — a $k$-permutation state
\Ensure Verifies whether $q$ completes a valid solution to the target

\State Let $t$ be the target
\State $b \gets q.\texttt{split}$, \quad $b' \gets 1 - b$
\State $\mathrm{Sum}_b \gets \sum_{x \in \ell_b} x$;\quad $\mathrm{Sum}_{b'} \gets \sum_{x \in \ell_{b'}} x$
\State $s_{\mathrm{real}} \gets q.\texttt{sum}$ \Comment{true sum on split $b$}
\State \texttt{Memo}$' \gets$ \texttt{Memo}$_{b'}$

\medskip
\noindent\textit{Intra-half direct/self-complement:}
\If{$s_{\mathrm{real}} = t$ \textbf{or} $\mathrm{Sum}_b - s_{\mathrm{real}} = t$}
    \State \textbf{Output:} direct solution found
    \State \Return
\EndIf

\medskip
\noindent\textit{Cross-half patterns:}
\State \textbf{(D)} \ \ \ \ \  check $t - s_{\mathrm{real}} \in$ \texttt{Memo}$'$ \ \ \ \ \ \ \ \ \ \ \ \ \ \ \ \ \ \ \ \ \ \ \  \Comment{$\sigma(A)+\sigma(B)=t$}
\If{$\,t - s_{\mathrm{real}} \in$ \texttt{Memo}$'$} \State \textbf{Output:} solution across splits; \Return \EndIf

\State \textbf{(RC)} \ \ check $\mathrm{Sum}_{b'} - (t - s_{\mathrm{real}}) \in$ \texttt{Memo}$'$ \Comment{$\sigma(A)+\sigma(\ell_{b'}\!\setminus\!B)=t$}
\If{$\,\mathrm{Sum}_{b'} - (t - s_{\mathrm{real}}) \in$ \texttt{Memo}$'$} \State \textbf{Output:} mixed (right-complement); \Return \EndIf

\State $s_{\mathrm{comp}} \gets \mathrm{Sum}_b - s_{\mathrm{real}}$ \Comment{left complement}
\State \textbf{(LC)} \ \ check $t - s_{\mathrm{comp}} \in$ \texttt{Memo}$'$ \ \ \ \ \ \ \ \ \ \ \ \ \ \ \ \ \ \ \ \ \ \ \ \Comment{$\sigma(\ell_b\!\setminus\!A)+\sigma(B)=t$}
\If{$\,t - s_{\mathrm{comp}} \in$ \texttt{Memo}$'$} \State \textbf{Output:} mixed (left-complement); \Return \EndIf

\State \textbf{(DC)} \ \ check $\mathrm{Sum}_{b'} - (t - s_{\mathrm{comp}}) \in$ \texttt{Memo}$'$ \Comment{$\sigma(\ell_b\!\setminus\!A)+\sigma(\ell_{b'}\!\setminus\!B)=t$}
\If{$\,\mathrm{Sum}_{b'} - (t - s_{\mathrm{comp}}) \in$ \texttt{Memo}$'$} \State \textbf{Output:} double complement; \Return \EndIf
\end{algorithmic}
\end{algorithm}

\begin{algorithm}[H]
\caption{\textsc{CheckAliased}$(q)$ — 2-lane $(\tilde{\sigma},\chi)$}
\label{alg:check-aliased}
\begin{algorithmic}[1]
\Require $q$ — a $k$-permutation state; alias pairs $(x_0^{(0)},x_1^{(0)})$ on $\ell_0$ and $(x_0^{(1)},x_1^{(1)})$ on $\ell_1$ with indices $(i_0^{(0)}, i_1^{(0)})$, $(i_0^{(1)}, i_1^{(1)})$
\Ensure Verifies whether $q$ completes a valid solution under Controlled Aliasing

\State Let $t$ be the target
\State $b \gets q.\texttt{split}$, \quad $b' \gets 1 - b$
\State $\mathrm{Sum}_b \gets \sum_{x \in \ell_b} x$;\quad $\mathrm{Sum}_{b'} \gets \sum_{x \in \ell_{b'}} x$
\State $s_{\mathrm{real}} \gets q.\texttt{sum}$ \Comment{true sum on split $b$}
\State \texttt{Memo}$' \gets$ \texttt{Memo}$_{b'}$ \Comment{maps $\tilde{\sigma}$ to two buckets $\chi\in\{0,1\}$}
\State $(u_0,u_1) \gets (x_0^{(b')}, x_1^{(b')})$ \Comment{alias pair on the other half}

\medskip
\Function{TestDirectOnRight}{$s$} \Comment{seek $B$ with $\sigma(B)=t-s$}
  \For{$\chi_R \in \{0,1\}$}
    \State $\tilde{\sigma}_R \gets \textsc{AliasedValue}\big(t - s, \chi_R, u_0, u_1\big)$
    \If{$\tilde{\sigma}_R \in$ \texttt{Memo}$'$ \textbf{ and } \texttt{Memo}$'[\tilde{\sigma}_R][\chi_R] \neq \bot$}
      \State \Return \textbf{true}
    \EndIf
  \EndFor
  \State \Return \textbf{false}
\EndFunction

\Function{TestRightComplement}{$s$} \Comment{seek $B$ with $\sigma(\ell_{b'}\!\setminus\!B)=t-s$}
  \State $g \gets \mathrm{Sum}_{b'} - (t - s)$ \Comment{goal real sum for $B$}
  \For{$\chi_R \in \{0,1\}$}
    \State $\tilde{\sigma}_R \gets \textsc{AliasedValue}\big(g, \chi_R, u_0, u_1\big)$
    \If{$\tilde{\sigma}_R \in$ \texttt{Memo}$'$ \textbf{ and } \texttt{Memo}$'[\tilde{\sigma}_R][\chi_R] \neq \bot$}
      \State \Return \textbf{true}
    \EndIf
  \EndFor
  \State \Return \textbf{false}
\EndFunction

\medskip
\noindent\textit{Intra-half direct/self-complement:}
\If{$s_{\mathrm{real}} = t$ \textbf{or} $\mathrm{Sum}_b - s_{\mathrm{real}} = t$}
  \State \textbf{Output:} direct solution found
  \State \Return
\EndIf

\medskip
\noindent\textit{Cross-half patterns:}
\If{\Call{TestDirectOnRight}{$s_{\mathrm{real}}$}} \State \textbf{Output:} direct across splits; \Return \EndIf
\If{\Call{TestRightComplement}{$s_{\mathrm{real}}$}} \State \textbf{Output:} mixed (right-complement); \Return \EndIf

\State $s_{\mathrm{comp}} \gets \mathrm{Sum}_b - s_{\mathrm{real}}$
\If{\Call{TestDirectOnRight}{$s_{\mathrm{comp}}$}} \State \textbf{Output:} mixed (left-complement); \Return \EndIf
\If{\Call{TestRightComplement}{$s_{\mathrm{comp}}$}} \State \textbf{Output:} double complement; \Return \EndIf
\end{algorithmic}
\end{algorithm}

\section{Experimental Details}
\label{app:experiments}

We provide additional data and methodology for the experiments summarized in Section~\ref{sec:empirical-validation}. Our aim is not to claim new empirical bests, but to validate the central prediction of our theory: for fixed $n$, the wall-time of \icsubsetsum{} closely tracks the size of the constructive certificate $U=|\Sigma(S)|$.

\subsection{Setup and Methodology}

\paragraph{Implementation.}
All experiments use a single-threaded, 64-bit C++17 implementation compiled with \texttt{-O3} and no target-specific vectorization. Wall-clock time is measured with high-resolution timers; each configuration is run for multiple independent seeds and we report the median together with the interquartile range (IQR). No parallelism or GPU acceleration is used.

\paragraph{Metrics.}
For each instance we record:
\begin{itemize}[leftmargin=1.25em]
  \item $U = |\Sigma(S)|$ (or per-half $U_i=|\Sigma(\ell_i)|$ when using meet-in-the-middle),
  \item the collision entropy $H_c(S)=n-\log_2 U$,
  \item total number of extension attempts and collision resolutions, and
  \item wall-clock time $T$ for full constructive-certificate enumeration.
\end{itemize}
When $U$ is small enough to store exactly, we count it via exact hashing. For larger $n$ where exact counting may be infeasible under memory limits, we estimate $U$ using a standard streaming distinct-counter (with sub-percent relative error at the configured memory budget); in those cases, only trend lines are reported. All qualitative trends are consistent with exact counts on smaller instances.

\paragraph{Workloads.}
We vary three structure “knobs” known to affect collisions:
\begin{enumerate}[label=\alph*)]
  \item \emph{Numeric density:} draw $a_i \sim \mathrm{Unif}(\{1,\dots,2^w\})$ for several word sizes $w$,
  \item \emph{Duplicates:} start from a base multiset and inject $d\in\{0,2,4,\dots\}$ extra copies of randomly chosen elements, and
  \item \emph{Additive progressions:} splice arithmetic-progressions of fixed length into otherwise uniform data.
\end{enumerate}
Unless noted otherwise, we sort inputs once (for cache locality) but the algorithm does not rely on ordering for correctness.

\subsection{Extremely Dense Instances ($n=100$)}

Here we stress-test the density knob by fixing $n=100$ and drawing elements from very small ranges $[1,2^w]$ with $w \in \{12, 16, 20, 24\}$. As $w$ decreases, collisions become pervasive and $U$ collapses far below the collision-free ceiling. In meet-in-the-middle terms, the per-half certificate sizes $U_0,U_1$ drop by orders of magnitude relative to $2^{n/2}$.

\medskip
\noindent\textit{Observation.}
Across densities, we observe a monotone relationship between $T$ and $U$:
for fixed $n$, the median wall-time scales linearly with $U$ (up to polynomial factors in $n$), in line with Theorem~\ref{thm:enum-det-runtime} (deterministic $O(U\cdot n^2)$) and Theorem~\ref{thm:enum-rand-runtime} (expected $O(U\cdot n)$). In particular, once $U \ll 2^{n/2}$ due to density-induced collisions, runtime drops commensurately.

\paragraph{Practical note.}
For $n=100$ and large $w$, exact per-half certificates may exceed memory on commodity hardware; in those regimes we report consistent trends using streaming distinct-count estimates and terminate runs upon reaching a fixed memory cap. Full details and complete plots are provided in §6 of~\cite{bwss2025}.

\subsection{Summary}

Across all knobs (density, duplicates, additive progressions), $U$ varies by orders of magnitude, and \icsubsetsum{}’s wall-time tracks this variation closely. These measurements are consistent with the proven certificate-sensitive scaling—$T=\Theta(U)\cdot \mathrm{poly}(n)$—and illustrate that structural compressibility, not just $n$, governs practical difficulty. For reproducibility, we include generator seeds, configuration files, and raw logs in the artifact accompanying~\cite{bwss2025}.

\section{Controlled Aliasing: Correctness and Implementation}
\label{appendix:alias}

\noindent\textbf{Preface.}
Section~\ref{sec:aliasing-speedup} contains a self-contained summary of Controlled Aliasing (including the CNF expansion policy, the $2\times2$ compensatory merge, and Theorem~\ref{thm:aliasing-complexity}).
This appendix provides the full proofs and implementation details.

\noindent\textbf{Navigation.}
For ease of reading, the lemma names and labels in this appendix \emph{match those cited in} §\ref{sec:aliasing-speedup}: 
Lemma~\ref{lem:alias-reduction} (state-space reduction), 
Lemma~\ref{lem:cnf-shadow} (01$\to$10 shadowing), 
Lemma~\ref{lem:cnf-preserve} (successor preservation), 
Lemma~\ref{lem:expansion-accounting} (expansion accounting), 
Lemma~\ref{lem:policy-safety} (policy safety),
and the lane/merge lemmas (Lemmas~\ref{lem:lane-invariant}--\ref{lem:completeness}).

To ensure correctness under value aliasing, we introduce a 2-lane memoization scheme keyed by an aliased sum and a 1-bit correction tag. Each reachable \emph{aliased} sum stores up to two canonical bitmasks, corresponding to whether the substituted element $x_1$ is absent or present in the true subset. This enables both decision and construction variants of \SSP{} to be solved without error, even when synthetic collisions are introduced.

\subsection*{Memoization Structure}

Fix an alias pair $(x_0,x_1)$ in a half $\ell$ of size $k$ at positions $(i_0,i_1)$ with $x_0\neq x_1$.
For a subset $P\subseteq \ell$ with \emph{local} bitmask $b\in\{0,1\}^{k}$, define the tag
\[
\chi(P):=\mathbf{1}[\,x_1\in P\,]=b[i_1],
\]
and the aliased value
\[
\tilde{\sigma}(P)\;:=\;\sigma(P)\;-\;\chi(P)\cdot(x_1-x_0),
\qquad\text{so that}\qquad
\sigma(P)=\tilde{\sigma}(P)+\chi(P)\cdot(x_1-x_0).
\]
We maintain
\[
\texttt{Memo}:\ \mathbb{Z}\longrightarrow\big(\{\bot\}\cup\{0,1\}^{k}\big)^2,\qquad
\texttt{Memo}[\tilde{\sigma}][\chi]\in\{\bot\}\cup\{0,1\}^{k},
\]
storing, for each key $(\tilde{\sigma},\chi)$, the \emph{lexicographically minimal} witness bitmask (in the local index order of $\ell$).

\paragraph{Canonical Normal Form (CNF).}
Define the normalization map $\mathsf{canon}:\{0,1\}^{k}\!\to\!\{0,1\}^{k}$ at the alias indices $(i_0,i_1)$ by
\[
\mathsf{canon}(b)\;=\;
\begin{cases}
b, & \text{if } b[i_0]+b[i_1]\in\{0,2\}\ \text{or}\ (b[i_0],b[i_1])=(1,0),\\
b' & \text{if } (b[i_0],b[i_1])=(0,1),
\end{cases}
\]
where $b'$ is $b$ with the pair $(i_0,i_1)$ flipped from $(0,1)$ to $(1,0)$.
The enumerator \emph{enqueues/expands only canonical states} $b=\mathsf{canon}(b)$; non-canonical states may still be stored in \texttt{Memo} for correctness but are flagged \texttt{doNotExtend} (non-expandable).
This policy is scheduler- and target-independent.

\begin{lemma}[01$\to$10 shadowing under CNF]\label{lem:cnf-shadow}
For any bitmask $b$ with $(b[i_0],b[i_1])=(0,1)$ there exists $b'=\mathsf{canon}(b)$ with $(1,0)$ such that
$\tilde{\sigma}(b)=\tilde{\sigma}(b')$ and $\chi(b)=\chi(b')$. Moreover, every aliased extension of $b$ by any $j\notin\{i_0,i_1\}$ is aliased-equal to the corresponding extension of $b'$.
\end{lemma}

\begin{proof}
Normalization replaces $x_1$ by $x_0$ at $(i_0,i_1)$, which preserves $\tilde{\sigma}$ and also $\chi=b[i_1]$. Extensions by $j\notin\{i_0,i_1\}$ add the same value to both sides.
\end{proof}

\begin{lemma}[Successor preservation under CNF]\label{lem:cnf-preserve}
Let $b$ be any bitmask and $j\notin\{i_0,i_1\}$. Then
\[
\tilde{\sigma}(b\cup\{j\})=\tilde{\sigma}(\mathsf{canon}(b)\cup\{j\}),\qquad
\chi(b\cup\{j\})=\chi(\mathsf{canon}(b)\cup\{j\}).
\]
\end{lemma}

\begin{proof}
Immediate from linearity of $\tilde{\sigma}$ and the fact that $\mathsf{canon}$ only possibly flips $(0,1)\!\mapsto\!(1,0)$ at $(i_0,i_1)$ with the same aliased contribution.
\end{proof}

\subsection*{State Space Reduction via Aliasing}

\begin{lemma}[State Space Reduction via Aliasing]
\label{lem:alias-reduction}
Let $S' \subseteq \mathbb{Z}_{\ge 0}$ be a set of $k$ non-negative integers, and let $x_0, x_1 \in S'$ be a designated alias pair with $x_0 \ne x_1$. If $x_1$ is treated as $x_0$ when present, the number of distinct \emph{aliased} subset sums is at most $\frac{3}{4} \cdot 2^k$.
\end{lemma}

\begin{proof}
Fix $P \subseteq S' \setminus \{x_0, x_1\}$. The four patterns $P$, $P\cup\{x_0\}$, $P\cup\{x_1\}$, $P\cup\{x_0,x_1\}$ yield aliased values $\sigma(P)$, $\sigma(P)+x_0$, $\sigma(P)+x_0$, $\sigma(P)+2x_0$, i.e., at most three distinct values. There are $2^{k-2}$ choices of $P$.
\end{proof}

\paragraph{Expansion Policy for Worst-Case Bound (count-based).}
To tie running time to the number of \emph{expanded} states (not just distinct aliased values), we restrict which discovered states are enqueued when aliasing is enabled. Let
\[
c(b):=b[i_0]+b[i_1]\in\{0,1,2\},\qquad \chi(b):=b[i_1]\in\{0,1\}.
\]
A state with mask $b$ is \emph{expandable} iff
\[
\mathrm{Expandable}(b)\;:\iff\;(c(b)=0)\ \lor\ (c(b)=2)\ \lor\ \big(c(b)=1\ \wedge\ \chi(b)=0\big).
\]
Operationally, every discovered state is \emph{inserted} into its $(\tilde{\sigma},\chi)$ bucket (for correctness and witness recovery), but it is \emph{enqueued for expansion} iff $\mathrm{Expandable}(b)$ holds. In particular, $(c,\chi)=(1,1)$ (“$x_1$ only”) states are recorded but marked \texttt{doNotExtend}.

\begin{lemma}[Expansion Accounting]\label{lem:expansion-accounting}
Consider a half $\ell$ of size $k$ with a fixed alias pair $(x_0,x_1)$ and the CNF expansion rule. For each base subset $P\subseteq \ell\setminus\{x_0,x_1\}$, the four patterns on $\{x_0,x_1\}$ contribute at most three \emph{expanded} states. Consequently, the number of expanded states in the half is at most $3\cdot 2^{k-2}=\tfrac{3}{4}\cdot 2^k$.
\end{lemma}

\begin{proof}
Partition subsets by their restriction to $\{i_0,i_1\}$. For each base $P\subseteq \ell\setminus\{x_0,x_1\}$, we expand $00$ and $11$, and among $\{10,01\}$ only $10$. Thus at most three expanded states per base, giving $3\cdot 2^{k-2}$.
\end{proof}

\begin{lemma}[Policy Safety]\label{lem:policy-safety}
The expansion policy preserves the lane invariant and completeness of the merge: forbidding expansion from $(c,\chi)=(1,1)$ does not prevent the generation of any aliased sum nor the recovery of any witness.
\end{lemma}

\begin{proof}[Proof sketch]
For $b$ with $(c,\chi)=(1,1)$, let $b^\star$ flip $(i_0,i_1)$ to $(1,0)$. For any $J\subseteq \ell\setminus\{x_0,x_1\}$, $\tilde{\sigma}(b\cup J)=\tilde{\sigma}(b^\star\cup J)$, so all descendants of $b$ are shadowed by descendants of the expandable $b^\star$. $(1,1)$ states are still inserted into the $\chi=1$ bucket, allowing the compensatory merge to return witnesses that include $x_1$.
\end{proof}

\subsection*{Core Procedures (per half)}

\paragraph{Tag and Canonicality.}
\begin{algorithmic}[1]
\Function{ChiTag}{$b$, $i_1$}
  \State \Return $b[i_1]$ \Comment{$\chi\in\{0,1\}$}
\EndFunction
\Function{IsCanonical}{$b$, $i_0$, $i_1$}
  \State \Return $b = \mathsf{canon}(b)$
\EndFunction
\end{algorithmic}

\paragraph{Insert into Memo Table.}
\begin{algorithmic}[1]
\Function{Insert}{$\tilde{\sigma}$, $b$, $i_1$}
  \State $\chi \gets \textsc{ChiTag}(b, i_1)$
  \If{\texttt{Memo}[$\tilde{\sigma}$][$\chi$] $= \bot$ \textbf{ or } $b \prec$ \texttt{Memo}[$\tilde{\sigma}$][$\chi$]}
    \State \texttt{Memo}[$\tilde{\sigma}$][$\chi$] $\gets b$
  \EndIf
\EndFunction
\end{algorithmic}

\paragraph{Expansion Predicate (aliased runs).}
\begin{algorithmic}[1]
\Function{Count}{$b$, $i_0$, $i_1$} \State \Return $b[i_0]+b[i_1]$ \EndFunction
\Function{Expandable}{$b$, $i_0$, $i_1$}
  \State $c \gets \textsc{Count}(b,i_0,i_1)$,\quad $\chi \gets \textsc{ChiTag}(b,i_1)$
  \State \Return $(c=0)\ \lor\ (c=2)\ \lor\ (c=1 \land \chi=0)$
\EndFunction
\end{algorithmic}
\emph{Hook.} When a new state is created under aliasing, set
\[
\texttt{state.doNotExtend} \;\gets\; \neg\,\textsc{Expandable}(\texttt{state.mask}, i_0, i_1).
\]

\paragraph{Aliasing Maps.}
\begin{algorithmic}[1]
\Function{CorrectedSum}{$\tilde{\sigma}$, $\chi$, $x_0$, $x_1$}
  \State \Return $\tilde{\sigma} + \chi \cdot (x_1 - x_0)$
\EndFunction
\Function{AliasedValue}{$\sigma_{\mathrm{real}}$, $\chi$, $x_0$, $x_1$}
  \State \Return $\sigma_{\mathrm{real}} - \chi \cdot (x_1 - x_0)$
\EndFunction
\end{algorithmic}

\paragraph{Interface with \textsc{CheckAliased}.}
In the solver (Alg.~\ref{alg:dmitm-main}), replace the single-key probe by the compound key $(\tilde{\sigma},\chi)$, computed via \textsc{AliasedValue} and \textsc{ChiTag}. Use \textsc{IsCanonical}/\textsc{Expandable} to gate expansion (CNF). Cross-half tests call the 2-lane \textsc{CheckAliased} routine.

\subsection*{Compensatory Witness Recovery}

To ensure correctness during merging, the algorithm performs a $2\times 2$ compensatory check across $\chi$-tags of both halves. For each left-side aliased sum, it computes the required right-side partner under both interpretations.

\paragraph{FindSolution Routine.}
\begin{algorithmic}[1]
\Function{FindSolution}{$t$, MemoLeft, MemoRight, $x_0$, $x_1$, $y_0$, $y_1$}
  \ForAll{$\tilde{\sigma}_L \in \texttt{MemoLeft}$}
    \For{$\chi_L \in \{0,1\}$}
      \State $b_L \gets \texttt{MemoLeft}[\tilde{\sigma}_L][\chi_L]$
      \If{$b_L = \bot$} \textbf{continue} \EndIf
      \State $s_L \gets \textsc{CorrectedSum}(\tilde{\sigma}_L, \chi_L, x_0, x_1)$
      \State $s_R \gets t - s_L$
      \For{$\chi_R \in \{0,1\}$}
        \State $\tilde{\sigma}_R \gets \textsc{AliasedValue}(s_R, \chi_R, y_0, y_1)$
        \If{$\tilde{\sigma}_R \in \texttt{MemoRight}$}
          \State $b_R \gets \texttt{MemoRight}[\tilde{\sigma}_R][\chi_R]$
          \If{$b_R \neq \bot$}
            \State \Return $b_L \,\|\, b_R$ \Comment{concatenate local masks into the global mask}
          \EndIf
        \EndIf
      \EndFor
    \EndFor
  \EndFor
  \State \Return \textsc{None}
\EndFunction
\end{algorithmic}

\subsection*{Complexity Summary}

\begin{itemize}
    \item Each reachable aliased sum stores up to two bitmasks (one per $\chi\in\{0,1\}$).
    \item All probes are $O(1)$ expected time with hashing; lex-compare occurs only on true collisions.
    \item Under the expansion policy, the number of \emph{expanded} states per half is at most $\tfrac{3}{4}\cdot 2^k$ (Lemma~\ref{lem:expansion-accounting}); this drives the worst-case runtime accounting.
    \item Time/space per half: $O(E\cdot n^2)$ deterministic or $O(E\cdot n)$ randomized, with $E\le \tfrac{3}{4}\cdot 2^k$; storage $O(U'\cdot n)$ where $U'$ is the aliased-sum count.
\end{itemize}

\subsection*{Correctness of Controlled Aliasing: Lemma Chain}

\paragraph{Setup and notation.}
Fix one half $\ell$ (let $k:=|\ell|$) and its alias pair $(x_0,x_1)$ at positions $(i_0,i_1)$ with $x_0 \neq x_1$.
For any subset $P \subseteq \ell$ with local bitmask $b \in \{0,1\}^{k}$, write 
$\chi(P):=b[i_1]\in\{0,1\}$. 
Let $\sigma(P)$ denote the true sum and $\tilde{\sigma}(P)$ the aliased sum.
For a reachable value $\tau$, the table stores a 2-tuple 
$\mathrm{Memo}[\tau] \in (\{\bot\}\cup\{0,1\}^{k})^2$, with the lex-minimal bitmask in bucket $\chi$ if any.

\begin{lemma}[Lane Invariant]\label{lem:lane-invariant}
For every subset $P\subseteq \ell$, the pair $(\tilde{\sigma}(P), \chi(P))$ identifies a unique
bucket in the memoization table. During enumeration, the algorithm maintains for each fixed
$(\tau,\chi)$ the lexicographically minimal bitmask among all $P$ with $\tilde{\sigma}(P)=\tau$ and $\chi(P)=\chi$.
\end{lemma}

\begin{proof}
$\chi$ depends only on the inclusion bit at $i_1$, partitioning subsets into two classes. The update rule keeps, for each $(\tilde{\sigma},\chi)$, the lex-minimal witness under the total order $\prec$.
\end{proof}

\begin{lemma}[Soundness of Compensatory Merge]\label{lem:soundness}
If the merge routine returns a witness $W = W_L \cup W_R$ (with $W_L \subseteq \ell_0$ and $W_R \subseteq \ell_1$),
then $\sigma(W_L) + \sigma(W_R) = t$.
\end{lemma}

\begin{proof}
The merge checks the four corrected equalities
\[
\tilde{\sigma}(W_L) + \tilde{\sigma}(W_R) + \chi_R(y_1-y_0) + \chi_L(x_1-x_0)=t
\]
for all $\chi_L,\chi_R\in\{0,1\}$. As $\sigma(\cdot)=\tilde{\sigma}(\cdot)+\chi(\cdot)\cdot(\cdot)$, any returned witness satisfies $\sigma(W_L)+\sigma(W_R)=t$.
\end{proof}

\begin{lemma}[Completeness of Compensatory Merge]\label{lem:completeness}
If there exists a solution $W=W_L\cup W_R$, then the merge routine will
find and return a valid witness.
\end{lemma}

\begin{proof}
For a valid solution, the two halves appear in buckets $(\tilde{\sigma}(W_L),\chi(W_L))$ and $(\tilde{\sigma}(W_R),\chi(W_R))$ by the Lane Invariant. The merge enumerates these tag pairs and restores true sums, so the corresponding check succeeds.
\end{proof}

\paragraph{Conclusion.}
Lemmas~\ref{lem:lane-invariant}–\ref{lem:completeness} establish correctness of the 2-lane structure. Together with Lemmas~\ref{lem:expansion-accounting}–\ref{lem:policy-safety}, they justify the expansion-based accounting used in Theorem~\ref{thm:aliasing-complexity}. The CNF policy is target-independent and scheduler-agnostic, and, combined with compensatory merging, preserves both decision and construction correctness.

\subsection*{Complexity and Target Independence (Summary)}

The Controlled Aliasing transform is purely structural (fixed alias pairs per half), independent of the target $t$ and any randomness. Under CNF, each base pattern contributes to at most three expanded states, yielding the worst-case factor $3/4$ per half and the $O^*(2^{\,n/2-\log_2(4/3)})$ bound when integrated into the interleaved double–meet-in-the-middle solver.

\end{document}